\documentclass[a4paper,fleqn]{cas-sc}

\usepackage[numbers]{natbib}

\usepackage{amsmath}
\usepackage{bm}
\usepackage{booktabs}
\usepackage{array}
\usepackage{tabularx}
\usepackage{xspace}
\usepackage[compat=0.6]{yquant}
\useyquantlanguage{groups}
\usepackage{makecell}
\usepackage{subcaption}
\usepackage{float}
\providecommand{\Description}[2][]{}

\makeatletter
\@ifundefined{insert@pcolumn}{%
  \let\insert@pcolumn\insert@column
}{}
\newcommand{\cas@enableHfloat}{%
  \let\cas@original@xfloat\@xfloat
  \def\@xfloat##1[##2]{%
    \edef\cas@floatpos{##2}%
    \ifdefstring{\cas@floatpos}{H}{%
      \@float@HH{##1}[H]%
    }{%
      \cas@original@xfloat{##1}[##2]%
    }%
  }%
}
\AddToHook{begindocument/end}{\cas@enableHfloat}
\makeatother

\newtheorem{definition}{Definition}
\newtheorem{theorem}{Theorem}
\newtheorem{lemma}{Lemma}

\newtheorem{corollary}{Corollary}
\newtheorem{example}{Example}
\newproof{proof}{Proof}
\AtEndEnvironment{proof}{%
  \unskip\nobreak\hfill\quad\hbox{$\Box$}%
}

\newcommand{\Ftwo}{\mathbb{F}_{2}}
\newcommand{\GL}{\operatorname{GL}}
\renewcommand{\O}{O}
\newcommand{\cnot}{\mathrm{CNOT}}
\newcommand{\NP}{\ensuremath{\mathsf{NP}}}
\newcommand{\MinCNOTv}{\ensuremath{\mathsf{MinCNOT}_{\mathrm{vanilla}}}\xspace}
\newcommand{\GridHP}{\textsf{Grid-HP}}
\renewcommand{\v}[1]{\boldsymbol{#1}}
\DeclareMathOperator{\rank}{rank}
\newcommand\<\langle
\renewcommand\>\rangle

\newenvironment{problemstatement}[1]{%
  \par
  \noindent%
  \tabularx{\linewidth}{@{}>{\normalfont\bfseries}l@{\quad}X@{}}
    Problem: & #1\\
}{%
  \endtabularx
  \par
}
\newcommand{\probleminstance}[1]{Instance: & #1\\}
\newcommand{\problemquestion}[1]{Question: & #1\\}

\pretocmd{\firstname}{\color{black}}{}{}

\begin{document}

\shorttitle{Vanilla Exact Synthesis of CNOT Circuits is NP-hard}
\shortauthors{C. Li and J. Guan}

\title[mode = title]{Vanilla Exact Synthesis of CNOT Circuits is NP-hard}
\tnotemark[1]

\tnotetext[1]{Dated: September 4, 2026; updated: September 27, 2026}

\author[1,2,3]{Chenjian Li}[orcid=0009-0001-7816-9784]

\ead{licj@ios.ac.cn}

\credit{Conceptualization and Investigation of this study, Methodology, Validation, Writing - Original draft preparation}

\author[1,2]{Ji Guan}[orcid=0000-0002-3490-0029]
\cormark[1]
\ead{guanji1992@gmail.com}
\credit{Funding Acquisition, Supervision, Writing – Review and Editing}

\cortext[cor1]{Corresponding author}

\affiliation[1]{organization={Key Laboratory of System Software (Chinese Academy of Sciences)},
                city={Beijing},
                postcode={100190},
                citysep={},
                country={China}}

\affiliation[2]{organization={Institute of Software, Chinese Academy of Sciences},
                city={Beijing},
                postcode={100190},
                citysep={},
                country={China}}

\affiliation[3]{organization={University of Chinese Academy of Sciences},
                city={Beijing},
                postcode={100190},
                citysep={},
                country={China}}

\begin{abstract}
Exact CNOT synthesis seeks a minimum-size circuit implementing a given invertible binary linear transformation. We study its vanilla formulation: synthesis from the identity on fixed labeled qubits, without ancillas and with all-to-all connectivity.
We prove that this problem is \NP-hard and that its bounded decision version is \NP-complete. Hardness persists even for unitriangular targets that are low-rank perturbations of the identity and under near-linear gate budgets.

Our proof gives a polynomial-time reduction from the \NP-complete Hamiltonian path problem on grid graphs. A unary hypercube embedding translates graph vertices into circuit parities and edge traversals into CNOT updates.
The central challenge is that Hamiltonian paths require intermediate vertex visits, whereas exact synthesis constrains only the final transformation. We bridge this gap with a replicated recorder qubits construction that encodes the required parities in the final outputs, together with a tight gate budget that forces every sufficiently short implementation to trace a Hamiltonian path. Crucially, this path structure is enforced by the construction rather than imposed as a restriction on the circuit.

The result directly implies \NP-hardness for shortest word problem and Cayley-graph distance computation over $\mathrm{GL}(n,2)$ with elementary transvections as generators, as well as fixed-width sequential XOR program minimization. For CNOT-minimal exact phase polynomial synthesis, it yields \NP-hardness with a zero phase polynomial; extending the reduction to grid-graph Hamiltonian cycle problem establishes hardness even when the final linear transformation is the identity. These complementary results show that the linear and phase components each independently suffice for \NP-hardness of exact phase polynomial synthesis.
\end{abstract}


\maketitle

\section{Introduction}
\label{sec:introduction}

The Controlled-NOT (CNOT) gate is one of the most important gates in quantum circuits and quantum computation. While single-qubit gates alone cannot synthesize general multi-qubit unitaries, adding CNOT gates makes universal quantum computation possible. As two-qubit gates typically have substantially higher error rates than single-qubit gates in the current NISQ era~\cite{nisq,zuchongzhi,willow,na}, its optimization is also particularly important for quantum hardware. Reducing the number of CNOT gates is therefore a key objective in quantum circuit synthesis and optimization,
%
%
and has been extensive studied.
Patel, Markov, and Hayes~\cite{pmh_optimal_linear_synthesis} first gave a block Gaussian elimination method (block GEM) that achieves the asymptotically optimal worst-case CNOT count $O(n^2/\log n)$ of CNOT circuits, where $n$ represents the number of qubits in the circuit. Subsequent work further improved practical synthesis through greedy algorithms~\cite{debrugiere2021Gaussian}, A* search~\cite{CNOT_opt_UCL,chen2025phasepoly}, and database-based method~\cite{CNOT_opt_group_theoretic}.

It is not until recent that the research community starts to focus on instance-wise optimality of CNOT synthesis, which seeks to find exactly the circuit(s) containing the minimum number of CNOT gates that implements a specified target linear transformation. Following the terminology of classical logic synthesis community, we will also refer to the instance-wise optimal synthesis of CNOT circuits as the \emph{exact synthesis} of CNOT circuits, or \emph{exact CNOT synthesis}.
Shaik and van de Pol encoded topology-aware exact CNOT synthesis into SAT instances and solved them using general-purpose SAT solvers~\cite{shaik2024Optimal}. More recently, Lin-search and the semi-tensor-product (STP) solver have attacked unrestricted exact CNOT synthesis using specialized search procedures, pruning rules, algebraic structure, and parallelization~\cite{li2026linsearch,li2026parallelizable}.  Despite these substantial domain-specific optimizations, the runtime cost of these exact synthesis methods grows exponentially with the problem size and the circuit length. This empirical behavior strongly suggests that exact CNOT synthesis is computationally intractable, but it does not by itself establish a complexity-theoretic hardness result.

Existing complexity results provide further evidence that exact CNOT synthesis is computationally difficult, but only for models with additional structure or assumptions. Kang showed that exact CNOT synthesis under restricted qubit connectivity is \NP-hard via a reduction from \textsf{VertexCover}~\cite{kang2023cnot}. Amy, Azimzadeh, and Mosca proved \NP-hardness for fixed-target parity networks (namely CNOT circuits) and for parity networks with linearly encoded inputs, the latter via a reduction from maximum-likelihood syndrome decoding~\cite{amy2018CNOTcomplexity}. A closely related problem, XOR minimization, has also been studied in the context of cryptography: Boyar, Matthews, and Peralta showed that the XOR minimization problem (\textsf{Shortest Linear Program}) is \NP-hard~\cite{boyar2013logic}; but such XOR programs may introduce arbitrary intermediate variables, whereas a CNOT circuit is a fixed-width, reversible computation in which the number of live variables remains unchanged. Hence, this hardness result does not directly apply to exact CNOT synthesis.

Taken together, these results provide strong evidence that exact CNOT synthesis is intrinsically difficult, but they leave an important conceptual gap.  Existing hardness proofs rely on additional structures---restricted connectivity, encoded inputs, or extra intermediate XOR variables---in the reduction, which are absent from the most elementary CNOT circuit model.  They therefore cannot rule out the possibility that the observed hardness originates from these auxiliary restrictions or resources rather than from the combinatorial complexity of CNOT synthesis itself.  The most vanilla formulation of exact CNOT synthesis asks for a minimum-size CNOT circuit on a fixed number of qubits (or variables), starting from the identity, using no ancillary qubits, and allowing a CNOT between every pair of qubits.  The complexity of this unrestricted in-place exact synthesis problem remained unresolved for many years after the question was raised by Amy et al.~\cite{amy2018CNOTcomplexity}.

In this work, we show that this unconditioned exact CNOT synthesis problem, which we denote by $\MinCNOTv^{opt}$, is indeed \NP-hard, and that its bounded decision version \MinCNOTv is \NP-complete.  Our proof gives a polynomial-time reduction from the \emph{Hamiltonian path problem on grid graphs} (\GridHP), which has been shown to be \NP-complete by Itai et al.~\cite{itai1982hamilton}.  Table~\ref{tab:hardness_comparison} summarizes our result together with previous hardness results for related CNOT and XOR synthesis models, and our result is shown as follows:

\vspace{1em}
\begin{problemstatement}{Vanilla CNOT minimization of CNOT circuits (\MinCNOTv)}
  \probleminstance{A number of qubits $n$, a target parity matrix $A\in\GL(n,2)$, and a nonnegative integer $k$.}
  \problemquestion{Provided that CNOT gates are available between any pair of qubits, does there exist an $n$-qubit CNOT circuit $C$ with length at most $k$ that realizes the linear transformation $A$?}
\end{problemstatement}
\vspace{-.6em}
\begin{theorem}
  \MinCNOTv is \NP-complete, and hence its optimization version $\MinCNOTv^{opt}$ is \NP-hard.
\end{theorem}

\begin{table}[width=.86\linewidth, pos=t]
\centering
\caption{Hardness results related to exact CNOT and XOR synthesis.}
\label{tab:hardness_comparison}

\begin{tabular}{ccc}
\toprule
\textbf{Work} &
\textbf{Additional structure/assumption} &
\textbf{Reduced from} \\
\midrule

\makecell{Kang and Ma~\cite{kang2023cnot}\\
(\textsf{Topology-Constrained CNOT Synthesis})}
&
\makecell{Restricted qubit\\connectivity}
&
\makecell{\textsf{Vertex Cover}}
\\
\hline
\makecell{Boyar et al.~\cite{boyar2013logic}\\
(\textsf{Shortest Linear Program})}
&
\makecell{Arbitrary intermediate\\XOR variables}
&
\makecell{\textsf{Vertex Cover}}\\
\hline
\makecell{Amy et al.~\cite{amy2018CNOTcomplexity}\\
(\textsf{Fixed-Target Parity Network})}
&
\makecell{Fixed CNOT target}
&
\makecell{\textsf{Hamming}\\ \textsf{Traveling Salesman}}
\\
\hline
\makecell{Amy et al.~\cite{amy2018CNOTcomplexity}\\
(\textsf{Encoded-Input Parity Network})}
&
\makecell{Encoded inputs}
&
\makecell{\textsf{Maximum-Likelihood}\\\textsf{Syndrome Decoding}}
\\

\hline

\makecell{\textbf{Ours}\\
(\textsf{Vanilla Exact CNOT Synthesis})}
&
\textbf{None}
&
\makecell{\textsf{Grid-Graph} \\ \textsf{Hamiltonian Path}}
\\

\bottomrule
\end{tabular}
\end{table}

Our proof proceeds in two steps. We first embed the grid graph onto a hypercube $\{0,1\}^d$ isometrically. We then reduce the Hamiltonian path problem on the hypercube to an exact CNOT synthesis problem instance, where vertices of the hypercube correspond to parities in a CNOT circuit, and edges correspond to CNOT gates. The main obstacle is that a Hamiltonian path requires every vertex to be visited intermediately, whereas exact CNOT synthesis constrains only the final matrix and cannot directly specify intermediate parities. We overcome this challenge by introducing a group of \emph{recorder qubits}, whose final outputs are required to contain (hence ``record'') the intermediate parities (namely vertices in the hypercube) that the circuit must reach. Together with a carefully designed tight gate budget, we show that a Hamiltonian path can be extracted from any CNOT circuit that both implements the prescribed final linear transformation and is sufficiently short, thus completing the reduction.

\paragraph{Main contribution.} To summarize, the main contributions of our work are as follows:

\begin{enumerate}
  \item We prove that exact CNOT synthesis is unconditionally \NP-hard via a reduction from \GridHP. To derive this result, we establish a reduction framework that combines the Hamiltonian path problem on grid graphs with a tight gate budget that forces the synthesized circuit to exhibit certain structures.

  \item We further prove the \NP-hardness of the shortest word problem and the Cayley graph distance problem over $\GL(n,2)$, as well as the minimization of s-XOR programs. These results demonstrate the broader implications of our hardness result for group theory and linear cryptography.

  \item By adapting the proof strategy within our reduction framework, we also prove that exact synthesis of phase polynomial circuits is \NP-hard, even when the prescribed final linear transformation is trivial or the phase polynomial is trivial. These results reveal two complementary and independently sufficient sources of hardness: either the target linear transformation or the phase polynomial alone can render exact synthesis of phase polynomial circuits intractable.
\end{enumerate}

\paragraph{Concurrent work.} Until very recently, the complexity of unrestricted, all-to-all CNOT exact synthesis remained open.  While preparing this manuscript, we noted the concurrent independent work by Acuaviva et al.~\cite{acuaviva2026cnot}, which established \NP-completeness for the same setting via a reduction from \textsf{VertexCover}.  Our result provides an independent and complementary proof based on a structurally different reduction from \GridHP. In particular, our reduction reveals a direct correspondence between CNOT parity trajectories and Hamiltonian paths in induced hypercube subgraphs, providing a structural perspective distinct from their \textsf{VertexCover}-based approach. 

\paragraph{Organization.} The remaining part of the manuscript is organized as follows. In Sec.~\ref{sec:background}, we give a brief introduction to the related topics, including CNOT circuits, exact CNOT synthesis, and the grid-graph Hamiltonian path problem. In Sec.~\ref{sec:reduction}, we present the reduction from \textsf{GridGraphHamilPath} to \MinCNOTv, proving the hardness of $\MinCNOTv$ and $\MinCNOTv^{opt}$, and further strengthen the hardness result under several structural restrictions on the target matrix and gate budget. Finally, in Sec.~\ref{sec:discussion} we discuss the broader complexity implications of our result, including shortest word problem over $\GL(n,2)$ and the related Cayley-graph distance problem, minimization of s-XOR programs. In particular, the hardness result for exact synthesis of phase polynomial circuits is discussed in Sec.~\ref{sec:phase_polynomial}.

\section{Background}
\label{sec:background}

\subsection{CNOT circuits and exact synthesis}
The CNOT gate, or the Controlled-NOT gate, is a two-qubit gate, which flips the target qubit (second qubit) if the control qubit (first qubit) is set:
\begin{align}
        &\cnot|00\>=|00\>, \qquad\cnot|01\>=|01\>,\\
        &\cnot|10\>=|11\>, \qquad\cnot|11\>=|10\>.
\end{align}
Alternatively, CNOT gates can be characterized as an addition-modulo-2 operation over the binary field $\Ftwo=\{0,1\}$:
    \begin{equation}
        \cnot|x_cx_t\>=|x_c(x_c\oplus x_t)\>
    \end{equation}
where $x_c,x_t\in\{0,1\}$, and ``$\oplus$'' denotes the addition modulo 2. As a result, a CNOT gate is \textit{linear}, in the sense that it satisfies the relation $f(x_1\oplus x_2)=f(x_1)\oplus f(x_2)$ for $x_1,x_2\in\mathbb F_2$ where $f(x)$ denotes $\cnot|x\>$.

We will call circuits that only contain CNOT gates \emph{CNOT circuits}. Although a quantum circuit generally requires a $2^n\times 2^n$ unitary matrix for its representation, a CNOT circuit implements a linear transformation over all $n$ input qubit variables, and can be compactly represented by an $n\times n$ parity matrix $A$. For example, a single CNOT gate can be represented as $\cnot_{1,2}\sim\left[\begin{smallmatrix} 1 & 0 \\ 1 & 1
    \end{smallmatrix}\right]$.
Meanwhile, CNOT gates are reversible operators, so CNOT circuits are also reversible, and therefore $A$ is invertible and belongs to the general linear group over $\Ftwo$: $A\in\GL(n,2)$.

\begin{example}\label{example:pmh_linear_circ}
    Consider the following CNOT circuit that was first introduced in~\cite{pmh_optimal_linear_synthesis}. Denote the input variables by $x_1$ to $x_4$. Then the circuit
    
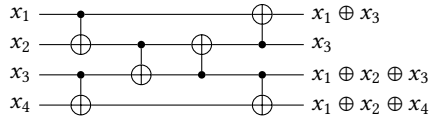
\vspace{-.8em}
    \begin{figure}[pos=H]
    \centering
    \begin{tikzpicture}
    \begin{yquant}[operator/separation=3mm]
        qubit {} q[4];
        init {$x_1$} q[0];
        init {$x_2$} q[1];
        init {$x_3$} q[2];
        init {$x_4$} q[3];
        cnot q[1]|q[0];
        cnot q[3]|q[2];
        cnot q[2]|q[1];
        cnot q[1]|q[2];
        cnot q[0]|q[1];
        cnot q[3]|q[2];
        output {$x_1\oplus x_3$} q[0];
        output {$x_3$} q[1];
        output {$x_1\oplus x_2\oplus x_3$} q[2];
        output {$x_1\oplus x_2\oplus x_4$} q[3];
      \end{yquant}
    \end{tikzpicture}
    \label{qcirc:GHZ_prepare}
    \end{figure}
    \vspace{-.6em}
    implements the following linear parity matrix:
\begin{equation}
    A=\begin{bmatrix}
        1 & 0 & 1 & 0 \\
        0 & 0 & 1 & 0 \\
        1 & 1 & 1 & 0 \\
        1 & 1 & 0 & 1 \\
    \end{bmatrix}.\label{eq:A_eg}
\end{equation}
\end{example}

At any point during the computation of a CNOT circuit, a qubit may contain a linear combination of the input variables such as $x_1\oplus x_3$. Following Amy's~\cite{amy2018CNOTcomplexity} notation, such expressions are called \emph{parities} and are denoted by
\begin{equation}
  \chi_{\v a}(\v x)
  := \v a^{\mathsf T}\v x
  = \bigoplus_{j=1}^{n} a_j x_j,
  \qquad \v a\in\Ftwo^n,
\end{equation}
where $\v a$ is the coefficient vector.

While CNOT circuits can be synthesized easily using Gaussian elimination methods, finding their smallest equivalent implementation, or exact CNOT synthesis, is much more challenging. The exact CNOT synthesis problem, and its bounded decision version, can be defined as follows:


\vspace{1em}
\begin{problemstatement}{Vanilla exact CNOT synthesis ($\MinCNOTv^{opt}$)}
  \probleminstance{A number of qubits $n$, a target parity matrix $A\in\GL(n,2)$.}
  \problemquestion{Provided that CNOT gates are available between any pair of qubits, what is a smallest CNOT circuit that implements the linear transformation $A$?}
\end{problemstatement}

\vspace{.6em}

\begin{problemstatement}{Vanilla minimization of CNOT circuit (\MinCNOTv)}
  \probleminstance{A number of qubits $n$, a target parity matrix $A\in\GL(n,2)$, and a nonnegative integer $k$.}
  \problemquestion{Provided that CNOT gates are available between any pair of qubits, does there exist an $n$-qubit CNOT circuit $C$ with length at most $k$ that realizes the linear transformation $A$?}
\end{problemstatement}
\vspace{.9em}

Apparently, $\MinCNOTv^{opt}\ge_{\rm poly} \MinCNOTv$. So it suffices to show the hardness of \MinCNOTv to demonstrate the hardness of the circuit synthesis problem. Furthermore, we can define the CNOT length of the parity matrix $A$ to be the smallest number of CNOT gates required to implement it:

\begin{definition}[CNOT length]
For $A\in\GL(n,2)$, its CNOT length is defined as
\begin{equation}
  \#\cnot(A)=\min\big\{|C|:
    C\,\, \text{\rm   is an ancilla-free $n$-qubit CNOT circuit that implements } A
  \big\}.
\end{equation}
\end{definition}

\begin{example}
  As shown in~\cite{li2026linsearch}, the smallest CNOT circuits implementing Eq.~\eqref{eq:A_eg} require only 5 CNOT gates. As a result, the CNOT length of~\eqref{eq:A_eg} is $\#\cnot(A)=5$. The smallest implementations are shown as follows:

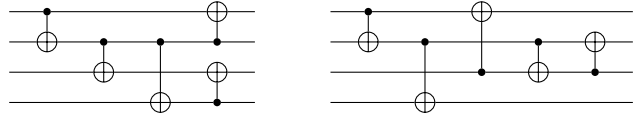
\begin{figure}[pos=H]
\centering
\begin{tikzpicture}
  \begin{yquantgroup}[operator/separation=2.5mm]
     \registers{
       qubit {} q[4];
     }
     \circuit{
       cnot q[1]|q[0];
       cnot q[2]|q[1];
       cnot q[3]|q[1];
       cnot q[2]|q[3];
       cnot q[0]|q[1];
     }
     \equals[\phantom{=}]
     \circuit{
       cnot q[1]|q[0];
       cnot q[3]|q[1];
       cnot q[0]|q[2];
       cnot q[2]|q[1];
       cnot q[1]|q[2];
     }
  \end{yquantgroup}
\end{tikzpicture}
\end{figure}
\end{example}

\subsection{Grid graphs and induced hypercube graphs}

\paragraph{Grid graphs.} A grid graph is a graph that takes grid points on a 2D integer grid as its vertices. Furthermore, vertices in a grid graph are connected if they correspond to neighboring grid points (Manhattan distance is one). That is, a grid graph contains exactly the unit edges between its vertices. Fig.~\ref{fig:grid_hamiltonian_path} gives an illustration of the grid graph, and grid graphs can be formally defined as follows:

\begin{definition}[Grid graph]
  A grid graph is a graph $G=(V_G, E_G)$ where $V_G\subseteq \mathbb Z^2$, and $E_G=\{(\v u,\v v): \v u,\v v\in V_G\text{ and }|\v u-\v v|=1\}$, where $|(x_u,y_u)-(x_v,y_v)|=|x_u-x_v|+|y_u-y_v|$ is the Manhattan distance.
\end{definition}

Moreover, we can consider the Hamiltonian path/cycle problems on grid graphs. Given distinct vertices $\v s$ and $\v t$, the Hamiltonian path problem on grid graphs ($\GridHP$) asks whether the grid graph has a Hamiltonian path from $\v s$ to $\v t$; similarly, the Hamiltonian cycle problem on grid graphs (\textsf{Grid-HC}) asks whether the grid graph has a Hamiltonian cycle.  Both problems have been shown by Itai et al. to be \NP-complete~\cite{itai1982hamilton}. \GridHP{} is used to show the hardness of exact CNOT synthesis (Theorem~\ref{thm:main}), whereas \textsf{Grid-HC} is used to show the hardness of exact phase polynomial synthesis when the target linear transformation is trivial (Theorem~\ref{thm:checkpoint_identity}).

\vspace{1em}
\begin{problemstatement}{Grid graph Hamiltonian path problem (\GridHP)}
  \probleminstance{A grid graph $G=(V_G,E_G)$, a starting point $\v s$ and ending point $\v t$.}
  \problemquestion{Does a Hamiltonian path of $G$ exist, with $\v s$ and $\v t$ being its endpoints?}
\end{problemstatement}

\begin{figure}[pos=t]
  \centering
  \begin{tikzpicture}[
    x=1.25cm,
    y=1.25cm,
    ambient edge/.style={gray!60,densely dotted,line width=0.65pt},
    graph edge/.style={black,line width=0.9pt},
    hamilton edge/.style={red!80!black,line width=2.2pt,
      line cap=round,line join=round},
    ambient vertex/.style={fill=gray!55,circle,inner sep=1.7pt},
    graph vertex/.style={fill=black,circle,inner sep=2.1pt}
  ]
    \foreach \y in {0,1,2}
      \draw[ambient edge] (-0.55,\y)--(2.55,\y);
    \foreach \x in {0,1,2}
      \draw[ambient edge] (\x,-0.55)--(\x,2.55);
    \foreach \x in {0,1,2}
      \foreach \y in {0,1,2}
        \node[ambient vertex] at (\x,\y) {};

    \foreach \xa/\ya/\xb/\yb in {
      0/0/1/0,1/0/2/0,
      0/1/1/1,1/1/2/1,
      0/2/1/2,
      0/0/0/1,0/1/0/2,
      1/0/1/1,1/1/1/2,
      2/0/2/1}
      \draw[graph edge] (\xa,\ya)--(\xb,\yb);

    \draw[hamilton edge]
      (0,2)--(0,1)--(0,0)--(1,0)--(2,0)--(2,1)--(1,1)--(1,2);

    \foreach \x/\y in {0/0,1/0,2/0,0/1,1/1,2/1,0/2,1/2}
      \node[graph vertex] at (\x,\y) {};
    \node[above left=1pt] at (0,2) {$s$};
    \node[above left=1pt] at (1,2) {$t$};

    \begin{scope}[shift={(3.15,1.65)}]
      \draw[ambient edge] (0,0)--(0.75,0);
      \node[ambient vertex] at (0.375,0) {};
      \node[anchor=west] at (0.95,0) {ambient grid};

      \draw[graph edge] (0,-0.55)--(0.75,-0.55);
      \node[graph vertex] at (0.375,-0.55) {};
      \node[anchor=west] at (0.95,-0.55) {grid graph $G$};

      \draw[hamilton edge] (0,-1.10)--(0.75,-1.10);
      \node[anchor=west] at (0.95,-1.10) {Hamiltonian path};
    \end{scope}
  \end{tikzpicture}
  \caption{Illustration of a grid graph $G$ over the 2D grid, together with an illustrative Hamiltonian path.}
  \Description{A finite square-grid window containing an induced grid graph.
  Vertices outside the graph and excluded grid edges are gray and dotted,
  vertices and edges of the graph are black, and a Hamiltonian path from s to t
  is highlighted by thick red edges.}
  \label{fig:grid_hamiltonian_path}
\end{figure}
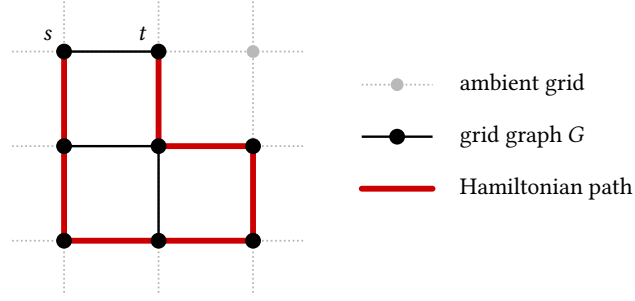

\paragraph{Hypercubes.} The $d$-dimensional hypercube graph $Q_d$, or $d$-dimensional hypercube for short, is a graph with vertex set $\Ftwo^d$. We also denote $|\v u-\v v|$ for the Manhattan distance between $d$-dimensional vectors, which happens to be the Hamming distance between hypercube vertices $\v u$ and $\v v$.  Thus two cube vertices are adjacent (and connected) when $|\v u-\v v|=1$. Furthermore, for $V\subseteq\Ftwo^d$, let $Q_d[V]$ denote the subgraph induced by $V$.
\begin{definition}[Induced hypercube subgraphs]
  For a vertex set $V\subseteq\Ftwo^d$, a hypercube subgraph induced by $V$ is defined as $Q_d[V]=(V, E)$ where $E=\{(\v u,\v v): \v u,\v v\in V\text{ and } |\v u-\v v|=1\}$, where $|\v u-\v v|$ is the Hamming distance.
\end{definition}

\section{The Main Hardness Result}
\label{sec:reduction}
In this section we prove the \NP-completeness of \MinCNOTv by reducing the grid graph Hamiltonian path problem to it. As a corollary, its optimization version $\MinCNOTv^{opt}$ is naturally \NP-hard. Then in Sec.~\ref{sec:strengthen}, we further adapt the reduction to strengthen the hardness results, showing that the hardness holds for a class of considerably simpler target linear transform, or even under when the synthesized circuit is not too deep.

\begin{theorem}
\label{thm:main}
$\MinCNOTv$ is \NP-complete.
\end{theorem}

\paragraph{Proof strategy.} Let $G$ be an input instance of \GridHP. Our reduction begins by mapping the vertices of $G$ isometrically into a hypercube $Q_d=\{0,1\}^d$ using a unary encoding of their grid coordinates.
The hypercube representation then naturally matches the algebra of CNOT circuits: a hypercube vertex $\v v=(v_1,\ldots,v_d)$ is represented by the parity
\begin{equation}
  \overline \chi_{\v v}(\v x)=x_0\oplus\bigoplus_{j=1}^{d}v_jx_j,
\end{equation}
where $\v x=(x_0,x_1,\ldots,x_d)$ denotes input variables of a CNOT circuit. If two vertices differ in coordinate $j$, their corresponding parities differ by exactly $x_j$. A CNOT gate that adds $x_j$ to the qubit carrying the current parity therefore realizes a traversal of the corresponding hypercube edge.

We then reduce the Hamiltonian path problem on the hypercube to an exact CNOT synthesis problem instance. The circuit to be synthesized contains a designated path qubit whose evolving parity represents the current position in the hypercube. A major obstacle, however, is that a Hamiltonian path requires every intermediate vertex to be visited, whereas in exact CNOT synthesis one can only constrain the final matrix and cannot directly prescribe intermediate parity checkpoints that must be reached. We overcome this challenge by introducing a group of \emph{recorder qubits}, whose final outputs are required to contain (hence \emph{record}) the intermediate parities (namely vertices in the hypercube) that the circuit must reach. In the intended circuit, the path qubit follows a Hamiltonian path through the embedded graph and, upon reaching each vertex, transfers its current parity to the associated recorder qubits.

A candidate circuit could nevertheless try to cheat and circumvent this requirement by mixing data and recorder qubits or by introducing temporary intermediate parities on the recorder qubits; to rule out such shortcuts, we require multiple replicas of the recorder qubits and impose a tight gate budget so that there is no room for such auxiliary computation. As a result, every sufficiently short circuit that meets the requirement would visit all required vertex parities, thereby inducing a Hamiltonian path and completing the reduction.

\begin{proof}
We first show that $\MinCNOTv\in\NP$.  Apparently, a CNOT circuit $C$ of length at most $k$ implementing the parity matrix $A$ can serve as a witness for a positive instance in $\MinCNOTv$. Note that each CNOT gate acts as an elementary row addition over $\Ftwo$, and can therefore be simulated on the parity matrix in $O(n)$ time.  Hence $C$ can be verified in $O(|C|\cdot n)$ time. After that, Patel, Markov, and Hayes~\cite{pmh_optimal_linear_synthesis} showed that every $n$-qubit linear reversible transformation admits a CNOT implementation of size at most $O(n^2/\log n)$.  Consequently, for each positive instance there is a witness of polynomial length that can be verified in $O(n^3/\log n)$ time.  That is to say, $\MinCNOTv\in\NP$.

We next prove \NP-hardness of $\MinCNOTv$ via a polynomial-time reduction from $\GridHP$.

\paragraph{\textbf{Step 1}: From grid graph to hypercube.} Let $(G,\v s,\v t)$ be a problem instance of $\GridHP$, where $G=(V_G,E_G)$ is the grid graph, and $\v s,\v t$ are the specified endpoints. Denote the number of vertices as $m:=|V_G|$. We first embed the grid graph $G$ into a hypercube $Q_d$.

For a vertex with associated coordinates $\v v_G=(x,y)\in V_G$, we define the unary encoding as:
\begin{equation}
  \mathrm{enc}_x(x)
  :=1^{\,x-x_{\min}}0^{\,x_{\max}-x}
  \in\Ftwo^{x_{\max}-x_{\min}},
\end{equation}
where $x_{\min},x_{\max}$ are the extremal horizontal coordinates for vertex coordinates in $G$, namely $x_{\min}=\min_{(x,y)\in V_G}x$ and $x_{\max}=\max_{(x,y)\in V_G}x$. Define $\mathrm{enc}_y(y)$ similarly for the vertical coordinate; then the entire vertex can be encoded as
\begin{equation}
  \mathrm{enc}(\v v_G)\equiv \mathrm{enc}(x,y):=\mathrm{enc}_x(x)\,\|\,\mathrm{enc}_y(y)
  \in\Ftwo^d,
\end{equation}
where $\|$ denotes string concatenation and $d=(x_{\max}-x_{\min})+(y_{\max}-y_{\min})$. As a slight abuse of notation, we do not distinguish the encoded binary strings and the tuples of binary coordinates, such as $110\equiv (1,1,0)$, as it should usually be clear from the context in our setting.

The unary encoding is isometric on both dimensions, hence the Manhattan distance on the original grid graph is preserved:
\begin{equation}
  |\mathrm{enc}(x,y)-\mathrm{enc}(x',y')|=|x-x'|+|y-y'|,\label{eq:isometry}
\end{equation}
where on the left-hand side the Manhattan distance happens to be the Hamming distance between binary strings as well.
As a result, all neighboring vertices in the original grid graph remain neighboring after the encoding, and the original edges are preserved in the hypercube. As a result, the proposed unary encoding scheme is an isomorphism.

Finally, we encode the specific problem instance $(G,\v s,\v t)$ into a Hamiltonian path problem instance $(V,\v s',\v t')$ on the hypercube graph $Q_d$. We define
\begin{align}
  V&=\{\mathrm{enc}(\v v_G)\oplus \mathrm{enc}(\v s):\v v_G\in V_G\},\\
  \v s'&=\mathrm{enc}(\v s)\oplus \mathrm{enc}(\v s)=\v 0,\\
  \v t'&=\mathrm{enc}(\v t)\oplus \mathrm{enc}(\v s).
\end{align}
And we have $m:=|V|=|V_G|$. Here, we have shifted every vertex by $\mathrm{enc}(\v s)$ to make the path start at $\v 0$. Then, $\v 0,\v t'\in V$ serve as the new endpoints in the hypercube. As Eq.~\eqref{eq:isometry} shows that the connections are preserved after the encoding, the induced subgraph on hypercube $Q_d[V]$ is isomorphic to $G$\footnote{The shifting also preserves distance, so the isomorphism is not broken after all vertices are shifted by $\oplus\mathrm{enc}(\v s)$.}.  Consequently,
\begin{equation}
  G\text{ admits a Hamiltonian path from }\v s\text{ to }\v t
  \quad\Longleftrightarrow\quad
  Q_d[V]\text{ admits a Hamiltonian path from }\v s'=\v 0\text{ to }\v t'.
  \label{eq:grid-cube-equivalence}
\end{equation}

The dimension is polynomially bounded. If $G$ is not connected, then the problem instance is trivial and can be immediately determined as negative without reduction. Since $G$ is connected, examining the minimal spanning tree of grid graphs shows that $d\le m-1$. Indeed, the spanning tree must contain at least $x_{\max}-x_{\min}$ horizontal edges to connect a leftmost vertex to a rightmost vertex, and at least $y_{\max}-y_{\min}$ vertical edges to connect a bottommost vertex to a topmost vertex. These two sets of edges are disjoint, so the tree contains at least $d$ edges. Since every spanning tree on $m$ vertices has exactly $m-1$ edges, we have $d\le m-1$.

\paragraph{\textbf{Step 2}: From hypercube Hamiltonian path to \MinCNOTv.} The hypercube formulation is particularly amenable to CNOT circuits: the vertices are represented as binary strings, and can be directly translated to parities on qubits in a CNOT circuit; walking along an edge on the hypercube is equivalent to flipping a vertex coordinate, and can be directly translated to a CNOT gate which flips the target qubit in the circuit.

A detailed demonstration of the reduction from hypercube Hamiltonian path to \MinCNOTv is shown in Fig.~\ref{fig:reduction-demonstration}.

\begin{figure}[pos=t]
  \centering
  \resizebox{\linewidth}{!}{%
  \begin{tikzpicture}[font=\footnotesize]
    \begin{yquant}[drawing mode=quality,operator/separation=0.4cm,register/separation=0.3cm]
      qubit {$x_0$} path;

      qubit {$x_1$} c1;
      qubit {$x_2$} c2;
      qubit {\raisebox{1mm}{\(\vdots\)}} cdots;
      setstyle {draw=none} cdots;
      qubit {$x_d$} cd;

      qubit {$y_{\v v_1,1}\sim y_{\v v_1,m}$} ru;
      setstyle {double,double distance=0.8pt} ru;
      qubit {$y_{\v v_2,1}\sim y_{\v v_2,m}$} rup;
      setstyle {double,double distance=0.8pt} rup;
      qubit {$y_{\v v_3,1}\sim y_{\v v_3,m}$} rthree;
      setstyle {double,double distance=0.8pt} rthree;
      qubit {\raisebox{1mm}{\(\vdots\)}} rdots;
      setstyle {draw=none} rdots;
      qubit {$y_{\v v_m,1}\sim y_{\v v_m,m}$} rupp;
      setstyle {double,double distance=0.8pt} rupp;

      {
        \yquantset{every control line/.style={draw,double,double distance=0.8pt}}
        [name=recordfirst]
        cnot ru | path;
      }

      [name=walkone]
      cnot path | c1;

      {
        \yquantset{every control line/.style={draw,double,double distance=0.8pt}}
        [name=recordsecond]
        cnot rthree | path;
      }

      [name=walkd]
      cnot path | cd;

      {
        \yquantset{every control line/.style={draw,double,double distance=0.8pt}}
        [name=recordthird]
        cnot rup | path;
      }

      align path,c1,c2,cdots,cd,ru,rup,rthree,rdots,rupp;
      hspace {3.0cm} path,c1,c2,cdots,cd,ru,rup,rthree,rdots,rupp;

      [name=pathout]
      output {$\overline\chi_{\v t'}(\v x)$} path;
      [name=c1out]
      output {$x_1$} c1;
      [name=c2out]
      output {$x_2$} c2;
      [name=cdotsout]
      output {\raisebox{1mm}{\(\vdots\)}} cdots;
      [name=cdout]
      output {$x_d$} cd;
      [name=ruout,font=\scriptsize]
      output {$y_{\v v_1,i}\oplus
        \overline\chi_{\v v_1}(\v x)$} ru;
      [name=rupout,font=\scriptsize]
      output {$y_{\v v_2,i}\oplus
        \overline\chi_{\v v_2}(\v x)$} rup;
      [name=rthreeout,font=\scriptsize]
      output {$y_{\v v_3,i}\oplus
        \overline\chi_{\v v_3}(\v x)$} rthree;
      [name=rdotsout]
      output {\raisebox{1mm}{\(\vdots\)}} rdots;
      [name=ruppout,font=\scriptsize]
      output {$y_{\v v_m,i}\oplus
        \overline\chi_{\v v_m}(\v x)$} rupp;
    \end{yquant}

    \node[anchor=south,align=center]
      at ([yshift=2mm]recordfirst-p)
      {$\overline\chi_{\v v_1}(\v x)=x_0$};

    \coordinate (circuitwest) at (current bounding box.west);
    \coordinate (bracecolumn) at ([xshift=0mm]circuitwest);
    \coordinate (annotationright) at ([xshift=-1.5mm]bracecolumn);
    \path (c1out) -- (cdout) coordinate[midway] (coordinatemid);
    \path (recordfirst) -- (ruppout) coordinate[midway] (recordermid);

    \node[anchor=east,text width=25mm,inner sep=0pt,
          align=center,font=\small]
      at ([xshift=3.5mm]annotationright |- recordsecond-p)
      {Path Qubit};

    \draw[decorate,decoration={brace,mirror,amplitude=4pt}]
      (bracecolumn |- c1out)
      -- (bracecolumn |- cdout);
    \node[anchor=east,text width=25mm,inner sep=0pt,
          align=center,font=\small]
      at ([xshift=3.5mm]annotationright |- coordinatemid)
      {Coordinate\\Qubits};

    \draw[decorate,decoration={brace,mirror,amplitude=4pt}]
      (bracecolumn |- recordfirst)
      -- (bracecolumn |- ruppout);
    \node[anchor=east,text width=25mm,inner sep=0pt,
          align=center,font=\small]
      at ([xshift=3.5mm]annotationright |- recordermid)
      {Recorder\\Qubits};

    \node[inner sep=0pt] at (walkone |- ruout) {$/$};
    \node[inner sep=0pt] at (walkone |- rupout) {$/$};
    \node[inner sep=0pt] at (walkone |- rthreeout) {$/$};
    \node[inner sep=0pt] at (walkone |- ruppout) {$/$};

    \node[inner sep=0pt]
      at (cdotsout -| walkone) {\raisebox{1mm}{\(\vdots\)}};
    \node[inner sep=0pt]
      at (rdotsout -| walkone) {\raisebox{1mm}{\(\vdots\)}};

    \path (recordfirst-p) -- (recordfirst)
      coordinate[midway] (verticalbundlemarkfirst);
    \node[inner sep=0pt,rotate=90] at (verticalbundlemarkfirst) {$/$};

    \path (recordsecond-p) -- (recordsecond)
      coordinate[midway] (verticalbundlemarksecond);
    \node[inner sep=0pt,rotate=90] at (verticalbundlemarksecond) {$/$};
    \path (recordthird-p) -- (recordthird)
      coordinate[midway] (verticalbundlemarkthird);
    \node[inner sep=0pt,rotate=90] at (verticalbundlemarkthird) {$/$};

    \path (recordthird-p) -- (pathout.west)
      coordinate[midway] (moregatesx);
    \node[inner sep=0pt,fill=white]
      at (pathout -| moregatesx) {$\cdots$};
    \node[inner sep=0pt,fill=white]
      at (c1out -| moregatesx) {$\cdots$};
    \node[inner sep=0pt,fill=white]
      at (c2out -| moregatesx) {$\cdots$};
    \node[inner sep=0pt,fill=white]
      at (cdotsout -| moregatesx) {\raisebox{1mm}{\(\vdots\)}};
    \node[inner sep=0pt,fill=white]
      at (cdout -| moregatesx) {$\cdots$};
    \node[inner sep=0pt,fill=white]
      at (ruout -| moregatesx) {$\cdots$};
    \node[inner sep=0pt,fill=white]
      at (rupout -| moregatesx) {$\cdots$};
    \node[inner sep=0pt,fill=white]
      at (rthreeout -| moregatesx) {$\cdots$};
    \node[inner sep=0pt,fill=white]
      at (rdotsout -| moregatesx) {\raisebox{1mm}{\(\vdots\)}};
    \node[inner sep=0pt,fill=white]
      at (ruppout -| moregatesx) {$\cdots$};

    \node[anchor=south,align=center] (checkpointlabel)
      at ([yshift=10mm]recordsecond-p.north)
      {$\overline\chi_{\v w_2}(\v x)
        =x_0\oplus x_1
        \sim(1,0,0,\ldots,0)$};
    \draw[-{Latex[length=1.8mm]}]
      (checkpointlabel.south) -- (recordsecond-p.north);
    \node[anchor=south,align=center]
      at ([xshift=12mm,yshift=2mm]recordthird-p.north)
      {$\overline\chi_{\v w_3}(\v x)
        =x_0\oplus x_1\oplus x_d
        \sim(1,0,\ldots,0,1)$};

  \end{tikzpicture}%
  }
  \caption{Demonstration of the reduction from hypercube Hamiltonian path to \MinCNOTv. $\v v_1(=\v 0),\v v_2,\ldots,\v v_m(=\v t')$ denote vertices in the hypercube. The first two walking steps and the first three recording steps are shown, where the vertices are visited in the order $\v 0-\v v_3-\v v_2-\cdots$. Here we have set $\v v_3=(1,0,\ldots,0,0)$ and $\v v_2=(1,0,\ldots,0,1)$ for demonstration. Each horizontal double line denotes the $m$ recorder qubits associated with one vertex, and each vertical double line denotes a bundle of $m$ CNOT gates targeting the same group of recorder qubits.
  }

  \label{fig:reduction-demonstration}
\end{figure}
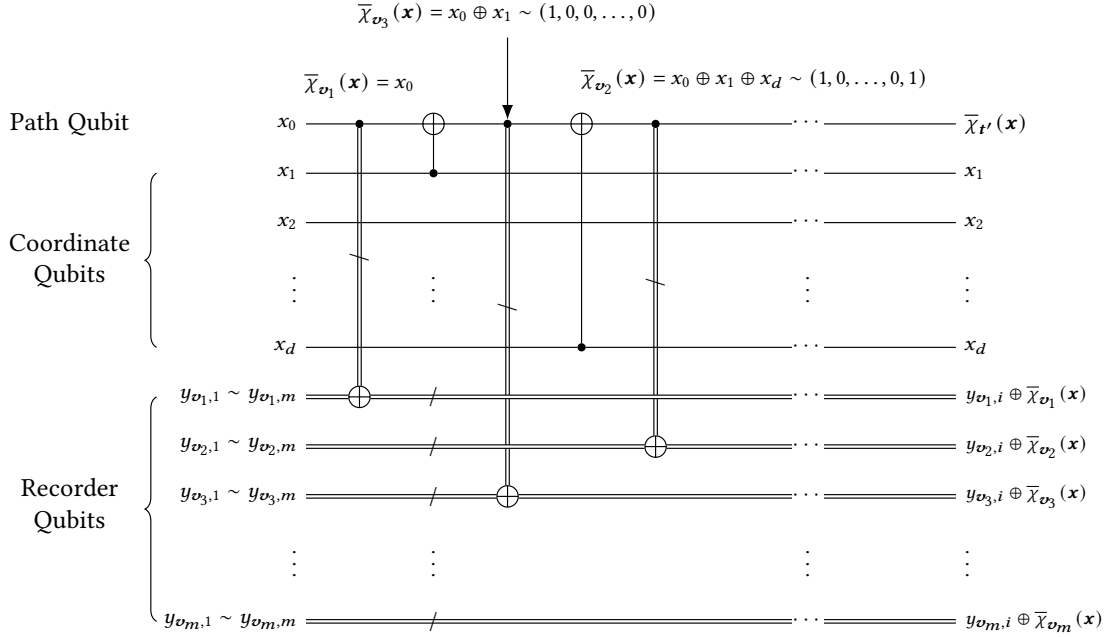

The construction is as follows. Firstly, we introduce three different types of qubits:
\begin{enumerate}
  \item 1 path qubit: a qubit that carries the parity whose value represents the current hypercube vertex in a candidate Hamiltonian path. The associated input variable of the path qubit is $x_0$.
  \item $d$ coordinate qubits: $d$ qubits that always carry the unit parities $x_1,\ldots, x_d$. They serve as the control qubits used by CNOT gates to walk the path qubit along an edge in the hypercube. The associated input variables of the coordinate qubits are $x_1,\ldots,x_d$.
  \item $m^2$ recorder qubits: $m$ groups of qubits, one group for each vertex $\v v\in V$, with $m$ qubits per group. Their outputs are required to contain the corresponding parity $\overline\chi_{\v v}(\v x)$, thereby forcing the circuit to reach each vertex $\v v$ during the computation. The associated input qubit variables of the recorder qubits are $y_{\v v,i},\, i=1,\ldots,m$ for each hypercube vertex $\v v\in V$.
\end{enumerate}
Then there are $n=d+1+m^2$ qubits altogether.

Now we first establish the correspondence between hypercube and CNOT circuits over the aforementioned qubits. Recall that each vertex $\v v=(v_1,\ldots,v_d)\in V$ is associated with the parity
\begin{equation}
    \overline\chi_{\v v}(\v x)=\chi_{1|\v v}(\v x)
    :=
    x_0\oplus\bigoplus_{j=1}^{d}v_jx_j,
    \qquad
    \v x=(x_0,x_1,\ldots,x_d).
    \label{eq:vertex_parity}
\end{equation}
Here, we have used $\overline\chi$ to denote the parities augmented with  the path variable $x_0$. In particular, $\overline\chi_{\v 0}(\v x)=x_0$. Moreover, for two neighboring vertices $\v u,\v v\in V$ that differ in the $j$-th coordinate, we have the following relation:
\begin{equation}
  |\v u-\v v|=1
    \quad\Longleftrightarrow\quad
  \v u\oplus\v v=\v e_j
    \quad\Longleftrightarrow\quad
  \overline\chi_{\v u}(\v x)\oplus \overline\chi_{\v v}(\v x)=x_j
    \quad\Longleftrightarrow\quad
  \overline\chi_{\v u}(\v x)\oplus x_j=\overline\chi_{\v v}(\v x)
    \label{eq:hypercube_parity_adj}
\end{equation}
where $\v e_j$ is the $j$-th standard basis vector of $\Ftwo^d$. Hence, whenever the path qubit carries $\overline\chi_{\v u}(\v x)$, a CNOT controlled by the $j$-th coordinate qubit and targeting the path qubit updates its parity to $\overline\chi_{\v v}(\v x)$, exactly corresponding to traversing the hypercube edge from $\v u$ to $\v v$.

The target linear transformation is then specified as follows:
\begin{alignat}{2}
    x_0 &\longmapsto \overline\chi_{\v t'}(\v x),\\
    x_j &\longmapsto x_j,
        &\qquad\quad & 1\le j\le d,\\
    y_{\v v,i} &\longmapsto y_{\v v,i}\oplus \overline\chi_{\v v}(\v x),
        & & \forall \,\v v\in V,\quad 1\le i\le m.
\end{alignat}
Furthermore, we let $A_V$ denote the parity matrix corresponding to the linear transformation over $\Ftwo$. Equivalently, $A_V$ has the following block form:
\begin{equation}
  A_V=
  \begin{bmatrix}
    1 & \v t' & \v 0 \\
    \v 0 & I_d & \v 0 \\
    \v 1 & X &  I_{m^2}
  \end{bmatrix},
  \label{eq:target-block}
\end{equation}
where $X\in\Ftwo^{m^2\times d}$ is a matrix formed by repeating each row $\v v\in V$ exactly $m$ times and combining all the $m$-row groups. Apparently, $\left[\begin{smallmatrix} 1 & \v t'\\ \v 0 & I_d \\
\end{smallmatrix}\right]$ is invertible;
Eq.~\eqref{eq:target-block} therefore implies $A_V$ is invertible and belongs to $\GL(n,2)$.
Finally, we set the gate bound to $k=m^2+m-1$,
where $m-1$ is designed as the budget for the steps in the Hamiltonian path, and $m^2$ is the budget for the recording.

To summarize, the reduction outputs an exact CNOT synthesis instance with the following parameters:
\begin{equation}
  n=d+1+m^2,\qquad
  A_V=\begin{bmatrix}
    1 & \v t' & \v 0 \\
    \v 0 & I_d & \v 0 \\
    \v 1 & X &  I_{m^2}
  \end{bmatrix},\qquad
  k=m^2+m-1.\label{eq:red_step2_res}
\end{equation}
Next, we show that
\begin{equation}
    Q_d[V]\text{ admits a Hamiltonian path between }\v 0\text{ and }\v t'
    \quad\Longleftrightarrow\quad
    (A_V,k)\in\MinCNOTv.
    \label{eq:step2_equivalence}
\end{equation}

\paragraph{\textbf{Completeness} ($\Rightarrow$): Constructing a CNOT circuit from a Hamiltonian path.}

This is the easy direction. Suppose
\begin{equation}
    \v w_1(=\v 0)-\v w_2-\ldots-\v w_{m}(=\v t')
\end{equation}
is a Hamiltonian path on $V$, where each $\v w_i$ vertex is distinct and is equal to some $\v v\in V$.  We construct a CNOT circuit satisfying the specification in Eq.~\eqref{eq:red_step2_res} as follows.

Initially, the path qubit carries $x_0=\overline\chi_{\v 0}(\v x)$.  Whenever the path qubit carries $\overline\chi_{\v w_i}(\v x)$, use this qubit as the control of $m$ CNOT gates, one targeting each recorder qubit associated with $\v w_i$.  This realizes
\begin{equation}
    y_{\v w_i,j}
    \longmapsto
    y_{\v w_i,j}\oplus\overline\chi_{\v w_i}(\v x),
    \qquad 1\le j\le m.
\end{equation}
Over all vertices, these recorder updates require $m^2$ CNOT gates in total.

For every consecutive vertex pair $\v w_i,\v w_{i+1}$, the Hamiltonian path edge implies that they differ in exactly one coordinate, say $\ell$.  A CNOT from the $\ell$-th coordinate qubit to the path qubit then updates
\begin{equation}
    \overline\chi_{\v w_i}(\v x)
    \longmapsto
    \overline\chi_{\v w_i}(\v x)\oplus x_\ell
    =
    \overline\chi_{\v w_{i+1}}(\v x).
\end{equation}
The $m-1$ edges of the path therefore require exactly $m-1$ additional CNOT gates.  At the end, the path qubit $0$ carries $\overline\chi_{\v t'}(\v x)$, and the remaining coordinate qubits still carry $x_1,\ldots,x_d$, and every recorder has the prescribed output.  Hence $A_V$ is implemented using exactly $m^2+m-1=k$ CNOT gates.

\paragraph{\textbf{Soundness} ($\Leftarrow$): Extracting a Hamiltonian path from a short CNOT circuit.}

Conversely, we now consider an arbitrary CNOT circuit $C$ of length $L:=|C|\le k$ that implements $A_V$. This direction is considerably more subtle. We introduced the concept  of path, coordinate, and recorder qubits during the reduction, but it is not a constraint imposed on the circuit synthesizer: an arbitrary implementation may mix these qubits, alter their intermediate parities, or use them as temporary workspace. Consequently, it is not trivial to extract the intended Hamiltonian path from an arbitrary circuit implementing $A_V$ with no more than $k$ CNOT gates. Fortunately, we can show, through some careful reasoning, that the intended division among the path, coordinate, and recorder qubits naturally appears in any valid CNOT circuit under the tight gate bound $k=m^2+m-1$.

We first note that all $m^2$ recorder qubits have changed after the circuit, so there are at least $m^2$ CNOT gates targeting the recorder qubits in $C$.

More precisely, collect the last CNOT gates targeting each recorder qubit into a set and denote it by $\mathcal G_{\rm record}$, and denote the  remaining gates by $\mathcal G_{\rm parity}$ (we will see very soon why these are the gates corresponding to the parities). Then we have:
\begin{equation}
  |\mathcal G_{\rm record}|= m^2,\quad |\mathcal G_{\rm parity}|=|C|-|\mathcal G_{\rm record}|=L-m^2\le m-1.\label{eq:gateset_def}
\end{equation}
We can then prove the following facts, which lead to the path, coordinate, and recorder qubit structure:

\emph{(1) All parities $\overline\chi_{\v v}(\v x)$ must appear in the circuit:} 
We show this fact by contradiction: if a parity $\overline\chi_{\v v}(\v x)$ never appears in the circuit, then it would cost at least $2m$ CNOT gates targeting the $m$ corresponding recorder qubits to fulfill their target transformation $y_{\v v, i}\oplus \overline\chi_{\v v}(\v x)$. This would result in at least $m(m-1)+2m=m^2+m>k$ CNOT gates in total, which is impossible. As a result, all $m$ parities must be present in the circuit at some point.

\emph{(2) $\mathcal G_{\rm parity}$ is precisely the set of gates that first generate the non-initial parities $\overline\chi_{\v v\neq\v 0}(\v x)$:} All the parities $\overline\chi_{\v v}(\v x)$ are absent at the beginning of the circuit except for $\overline\chi_{\v 0}(\v x)$, leaving $m-1$ distinct parities to be generated. It therefore takes at least another $m-1$ CNOT gates to generate them. These parity-generating CNOT gates cannot be the last gates targeting recorder qubits (and hence cannot belong to $\mathcal G_{\rm record}$), because they result in parities $\overline\chi_{\v v}(\v x)$, but not the required output $y_{\v v',i}\oplus \overline \chi_{\v v'}(\v x)$ on any recorder qubits. In other words, the CNOT gates that generate new parities can only be in $\mathcal G_{\rm parity}$. As a result, we have $|\mathcal G_{\rm parity}|\ge m-1$. Combining with Eq.~\eqref{eq:gateset_def}, we see that
\begin{equation}
  m-1\le |\mathcal G_{\rm parity}|\le m-1\implies |\mathcal G_{\rm parity}|=m-1.
\end{equation}
The relation above shows that there is no gate surplus in $\mathcal G_{\rm parity}$, and every CNOT gate in $\mathcal G_{\rm parity}$ must derive a new non-initial parity $\overline\chi_{\v v\neq\v 0}(\v x)$ on its target qubit. As a result, there is a one-to-one correspondence between the gates in $\mathcal G_{\rm parity}$ and all non-initial parities $\overline\chi_{\v v\neq\v 0}(\v x)$.

\emph{(3) Gates in $\mathcal G_{\rm parity}$ cannot target recorder qubits:} Note that every gate in $\mathcal G_{\rm parity}$ leaves a parity $\overline\chi_{\v v}(\v x)$ on its target qubit, which does not contain the recorder variables $y_{\v v,i}$. On the other hand, gates in $\mathcal G_{\rm record}$ only target recorder qubits. Combining both facts, we see that the recorder variables $y_{\v v,i}$ are never injected into the path qubit or any coordinate qubits. As the path and coordinate qubits do not contain recorder variables $y_{\v v,i}$ at the beginning, they can never contain the recorder variables $y_{\v v,i}$ at any location in the circuit.

If a gate in $\mathcal G_{\rm parity}$ targets a recorder qubit, then the targeted recorder qubit would carry $\overline\chi_{\v v}(\v x)$ and its original recorder variable $y_{\v v',i}$ would be punched out. Consequently, $m^2$ recorder variables would have to be carried exclusively by the remaining $m^2-1$ recorder qubits, as the path qubit and coordinate qubits never carry recorder variables. This fact contradicts the reversibility requirement of CNOT circuits. Consequently, gates in $\mathcal G_{\rm parity}$ cannot target recorder qubits.

\emph{(4) Gates in $\mathcal G_{\rm parity}$ cannot target coordinate qubits:} If a gate in $\mathcal G_{\rm parity}$ targets a coordinate qubit, then, as shown by (2), it leaves a parity $\overline\chi_{\v v}(\v x)$ on the target qubit. Note that $\overline\chi_{\v v\neq\v 0}(\v x)\neq x_j$ for any coordinate $x_j$, so at least one extra gate is needed to recover the required output $x_j$ on the coordinate qubit. However, no gate can fulfill this task: gates in $\mathcal G_{\rm record}$ do not target coordinate qubits by definition, and we have shown that gates in $\mathcal G_{\rm parity}$ must leave $\overline\chi_{\v v\neq\v 0}(\v x)$ on their target qubits, not $x_j$. By contradiction, gates in $\mathcal G_{\rm parity}$ cannot target coordinate qubits.

\vspace{.7em}
(2),(3) and (4) tell us that all gates in $\mathcal G_{\rm parity}$ can only target the path qubit, and each of them generates a new non-initial parity $\chi_{\v v\neq\v 0}(\v x)$ on its target qubit. On the other hand, all other $m^2$ gates are in $\mathcal G_{\rm record}$, each targeting a recorder qubit. To summarize, the properties proved above can be organized into the following lemma:
\begin{lemma}[Gate structure]\label{lemma:gate_struct}
  Suppose $C$ is a short circuit with $L:=|C|\le k$ CNOT gates implementing the problem instance specified by~\eqref{eq:red_step2_res}. Then:
  \begin{enumerate}
    \item the CNOT gates in $C$ can be divided into two disjoint gate sets $\mathcal G_{\rm parity}$ and $\mathcal G_{\rm record}$;
    \item $|\mathcal G_{\rm parity}|=m-1$, and all gates in $\mathcal G_{\rm parity}$ target the path qubit and leave a new non-initial parity $\overline\chi_{\v v\neq\v 0}(\v x)$ on the path qubit;
    \item $|\mathcal G_{\rm record}|=m^2$, and every gate in $\mathcal G_{\rm record}$ targets one of the $m^2$ recorder qubits and changes it from $y_{\v v,i}$ to the target $y_{\v v,i}\oplus \overline\chi_{\v v}(\v x)$.
  \end{enumerate}
\end{lemma}

From Lemma~\ref{lemma:gate_struct}, we see that, starting from $x_0=\overline\chi_{\v w_1=0}(\v x)$, the path qubit is targeted only by the $m-1$ parity-generating CNOT gates in $\mathcal G_{\rm parity}$, which leave $m-1$ consecutive parities $\overline\chi_{\v w_2}(\v x), \overline\chi_{\v w_3}(\v x), \ldots, \overline\chi_{\v v_{m}=\v t'}(\v x)$ on the path qubit. From Eq.~\eqref{eq:hypercube_parity_adj}, we see that all the vertices corresponding to the parities $\v 0,\v w_2,\v w_3,\ldots,\v t'$ are adjacent on the hypercube. We can then naturally extract them and derive a Hamiltonian path $\v 0-\v w_2-\v w_3-\cdots-\v t'$.

Now we have completed the proof.
\end{proof}

As a direct result, the bounded decision version of the problem can be reduced to the optimization version of the problem. As a corollary, we immediately derive the following hardness result:

\begin{corollary}
  $\MinCNOTv^{opt}$ is \NP-hard.
\end{corollary}

\subsection{Further strengthening of the hardness theorem}
\label{sec:strengthen}

In the reduction from $\GridHP$ to \MinCNOTv, only a specific type of linear transformation is used in Eq.~\eqref{eq:red_step2_res}. In this section we show how to exploit the fact and strengthen our hardness result. In particular, we have the following strengthened version of Theorem~\ref{thm:main}:

\begin{theorem}\label{thm:main_strengthen}
  \MinCNOTv is \NP-complete, and $\MinCNOTv^{opt}$ is \NP-hard. Furthermore, the \NP-completeness and \NP-hardness hold even for problem instances $(A,k)$ in which
  \begin{enumerate}
    \item the target linear transformation $A$ is unitriangular;
    \item $A$ satisfies $(A-I)^3=0$;
    \item $\rank(A-I)=\O(\sqrt{n})$;
    \item $k<n+\O(\sqrt{n})$.
  \end{enumerate}
\end{theorem}
\begin{proof}
  We modify the proof of Theorem~\ref{thm:main} to show (1) and (2). In the reduction from $\GridHP$ to \MinCNOTv, reorder the qubits into (i) coordinate qubits, (ii) the path qubit, and (iii) recorder qubits. Then the input variable sequence becomes
  \begin{equation}
    x_1,\ldots,x_d,x_0,y_{\v v_1,i},\ldots,y_{\v v_m,i},
  \end{equation}
  and the constructed target linear transformation $A$ becomes
  \begin{equation}
  A_V=\begin{bmatrix}
    I_d & \v 0 & \v 0 \\
    \v t' & 1 & \v 0 \\
    X & \v 1 &  I_{m^2}
  \end{bmatrix},
  \end{equation}
  which is lower unitriangular. Thus it suffices to consider problem instances in \MinCNOTv where target linear transformations $A$ are lower unitriangular to make the problem intractable. The same reasoning also holds for upper unitriangular matrices: if we further exchange the position of the recorder qubits and the coordinate qubits, then we obtain target linear transformations that are upper unitriangular. As a result, we have shown the first strengthening statement of the theorem.

  Then direct computation of $A-I$ can show the second strengthening statement\footnote{Note $\v t'$ is a row vector.}:
  \begin{equation}
    A_V-I=\begin{bmatrix}
    \v 0 & \v 0 & \v 0 \\
    \v t' & 0 & \v 0 \\
    X & \v 1 & \v 0
  \end{bmatrix},\quad
  (A_V-I)^2=\begin{bmatrix}
    \v 0 & \v 0 & \v 0 \\
    \v 0 & 0 & \v 0 \\
    \v 1\ \v t' & \v 0 & \v 0
  \end{bmatrix},\quad
  (A_V-I)^3=\v 0,
  \end{equation}
  where
  \begin{equation}
    \v 1\ \v t'=\begin{bmatrix}
      1 \\ \vdots \\ 1
    \end{bmatrix}
    \cdot (t'_1,\cdots,t'_d)
    =
    \begin{bmatrix}
      - & \v t' & - \\
      & \vdots \\
      - & \v t' & -
    \end{bmatrix}
  \end{equation}
  is a $m^2\times d$ matrix that contains $m^2$ copies of $\v t'$ as its rows. So it suffices to consider problem instances in \MinCNOTv in which $(A-I)^3=0$ to make the problem intractable.

  The third and fourth strengthening statements require examining the relation between $d,m$ and $n$. Recall that $k=m^2+m-1$, $n=d+1+m^2$, and $0\le d\le m-1$. Then we have:
  \begin{equation}
    m^2+1\le n=d+1+m^2 \le m^2+m<(m+1)^2\implies m=\lfloor \sqrt{n}\rfloor.
  \end{equation}
  As a result, $\rank(A_V-I)\le \rank(\v t')+\rank(X)\le m=O(\sqrt{n})$, demonstrating the third strengthening statement.  That is to say, \NP-completeness persists even when the target matrix is a low-rank perturbation of the identity.

  Finally, using the relation $m=\lfloor\sqrt{n}\rfloor$, we see that for the $k$ used in the reduction we have the following relation:
  \begin{equation}
    k=m^2+m-1=(\lfloor\sqrt{n}\rfloor)^2+\lfloor\sqrt{n}\rfloor-1\le n+\O(\sqrt{n}),
  \end{equation}
  thus completing the fourth strengthening statement in the theorem.
\end{proof}

Theorem~\ref{thm:main_strengthen} shows that the hardness even holds for relatively simple exact CNOT synthesis problem instances. First, the target transformation is algebraically close to the identity: $A$ is unitriangular, and is a low-rank perturbation to the identity. Second, the hardness even persists for relatively shallow CNOT circuits: while the worst-case CNOT complexity can reach $\Theta(n^2/\log n)$ for an $n$-qubit CNOT circuit, the instances produced by our reduction satisfy $k<n+\O(\sqrt n)$. Hence, the computational difficulty is already present in the near-linear gate-count regime, far below the worst-case synthesis complexity.

The structural strengthenings obtained here and in the concurrent work of Acuaviva et al.~\cite{acuaviva2026cnot} are broadly comparable but emphasize different aspects of the problem. Both results establish hardness for unitriangular transformations satisfying $(A-I)^3=0$, while  Acuaviva et al. further derive sparsity restrictions and approximation hardness. Our construction sharpens two other parameters: Theorem~\ref{thm:main_strengthen} gives hardness even when $\rank(A-I)=\O(\sqrt n)$ and under tighter budget $k<n+\O(\sqrt n)$, making the two proofs complementary in both technique and structural implications.

\section{Discussion and Complexity Implications}\label{sec:discussion}

The simplicity of \MinCNOTv allows its complexity implications extend beyond CNOT circuit synthesis itself and applicable to many related problems as well. With no ancillary qubits, topology constraints, or other auxiliary structure, the problem amounts to finding a shortest sequence of elementary linear updates that realizes a prescribed linear invertible transformation over $\Ftwo$. Consequently, the same optimization problem admits several equivalent formulations in algebra, graph theory, linear cryptography, and also appears as a special case of more general quantum-circuit synthesis problems. In this section, we derive complexity implications for shortest words in $\GL(n,2)$, distances in the corresponding Cayley graph, shortest s-XOR programs, and exact synthesis of phase polynomial circuits.

\subsection{Shortest Word Problem}\label{sec:shortest-word}

The most immediate algebraic interpretation of \MinCNOTv is as a shortest word problem over the general linear group $\GL(n,2)$. Recall that each CNOT gate corresponds to an elementary transvection
\begin{equation}
\cnot_{j,i}\sim T_{i,j}=I_n+\v e_i\v e_j^{\mathsf T},\qquad i\neq j,
\end{equation}
whose left multiplication performs the row addition of the matrix $R_i\leftarrow R_i+R_j$. Hence, the set
\begin{equation}
\mathcal T_n:=\{T_{i,j}:i,j\in[n],\,i\neq j\}
\end{equation}
generates $\GL(n,2)$, and a CNOT circuit implementing $A\in\GL(n,2)$ directly corresponds to a word over $\mathcal T_n$ whose product equals $A$.

For a group $G$ generated by a set $S$, the \emph{shortest word problem} of $g\in G$ asks what is the shortest sequence of elementary operations in $S$ needed to generate $g$. In particular, a \emph{word} over $S$ is a finite sequence $s_\ell\cdots s_2s_1$ with $s_i\in S$, and the word length of $g\in G$ with respect to $S$ is
\begin{equation}
\ell_{S}(g):=\min\left\{\ell:g=s_\ell\cdots s_1,\ s_i\in S\right\}.
\end{equation}
The corresponding shortest word problem then asks, given $g$ and an integer $k$, whether $\ell_{S}(g)\le k$.

For $G=\GL(n,2)$ and $S=\mathcal T_n$, the word length $\ell_{\mathcal T_n}(A)$ is precisely the minimum number of CNOT gates required to synthesize $A$. Therefore, \MinCNOTv directly corresponds to the shortest word problem for $\GL(n,2)$ under the generating set of elementary transvections.

\begin{corollary}
The shortest word problem over $\GL(n,2)$ with respect to the elementary transvections $\mathcal T_n$ is \NP-complete. Consequently, computing $\ell_{\mathcal T_n}(A)$ is \NP-hard.
\end{corollary}

\subsection{Cayley Graph Distance}

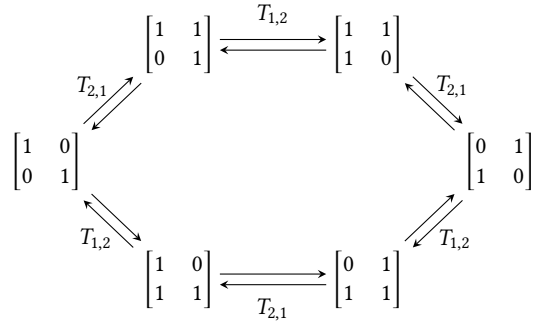
\begin{figure}[pos=t]
\centering
\begin{tikzpicture}[scale=1.0, every node/.style={inner sep=2pt}, >=stealth]

\node (I) at (-3,0) {$\begin{bmatrix}1&0\\0&1\end{bmatrix}$};
\node (A) at (-1.25,1.55) {$\begin{bmatrix}1&1\\0&1\end{bmatrix}$};
\node (C) at (1.25,1.55) {$\begin{bmatrix}1&1\\1&0\end{bmatrix}$};
\node (P) at (3,0) {$\begin{bmatrix}0&1\\1&0\end{bmatrix}$};
\node (D) at (1.25,-1.55) {$\begin{bmatrix}0&1\\1&1\end{bmatrix}$};
\node (B) at (-1.25,-1.55) {$\begin{bmatrix}1&0\\1&1\end{bmatrix}$};

\draw[->] ([xshift=-1.5pt,yshift=1.2pt]I.north east) -- ([xshift=-1.5pt,yshift=1.2pt]A.south west);
\draw[->] ([xshift=1.5pt,yshift=-1.2pt]A.south west) -- ([xshift=1.5pt,yshift=-1.2pt]I.north east);
\path (I) -- (A) node[pos=.84,left=6pt] {$T_{2,1}$};

\draw[->] ([yshift=2pt]A.east) -- ([yshift=2pt]C.west);
\draw[->] ([yshift=-2pt]C.west) -- ([yshift=-2pt]A.east);
\path (A) -- (C) node[midway,above=5pt] {$T_{1,2}$};

\draw[->] ([xshift=1.5pt,yshift=1.2pt]C.south east) -- ([xshift=1.5pt,yshift=1.2pt]P.north west);
\draw[->] ([xshift=-1.5pt,yshift=-1.2pt]P.north west) -- ([xshift=-1.5pt,yshift=-1.2pt]C.south east);
\path (C) -- (P) node[pos=.14,right=6pt] {$T_{2,1}$};

\draw[->] ([xshift=1.5pt,yshift=-1.2pt]P.south west) -- ([xshift=1.5pt,yshift=-1.2pt]D.north east);
\draw[->] ([xshift=-1.5pt,yshift=1.2pt]D.north east) -- ([xshift=-1.5pt,yshift=1.2pt]P.south west);
\path (P) -- (D) node[pos=.1,below=8pt] {$T_{1,2}$};

\draw[->] ([yshift=-2pt]D.west) -- ([yshift=-2pt]B.east);
\draw[->] ([yshift=2pt]B.east) -- ([yshift=2pt]D.west);
\path (D) -- (B) node[midway,below=5pt] {$T_{2,1}$};

\draw[->] ([xshift=-1.5pt,yshift=-1.2pt]B.north west) -- ([xshift=-1.5pt,yshift=-1.2pt]I.south east);
\draw[->] ([xshift=1.5pt,yshift=1.2pt]I.south east) -- ([xshift=1.5pt,yshift=1.2pt]B.north west);
\path (B) -- (I) node[pos=.1,left=6pt] {$T_{1,2}$};

\end{tikzpicture}
\caption{The Cayley graph $C_G$ for $G=\GL(2,2)$ generated by $S=\mathcal T_2=\{T_{1,2},T_{2,1}\}$.}
\label{fig:cayley-gl22}
\end{figure}

The shortest word problem formulation has an immediate graph theoretic interpretation in terms of \emph{Cayley graphs}. Let $G=\langle S\rangle$ be a finite group generated by $S$. The Cayley graph of $G$ with respect to $S$ is the graph $C_G=(V,E)$ with
\begin{equation}
V=G,\qquad E=\{(g,h): g,h\in G,\ h=sg\text{ for some }s\in S\}.
\end{equation}
Thus, group elements form the vertices of $C_G$, while actions of the generators  correspond to the edges. The Cayley graph of $\GL(2,2)$ with respect to all transvections $\mathcal T_2$ is shown in Fig.~\ref{fig:cayley-gl22}.

For $g,h\in G$, a natural question is, what is the smallest number of actions that can transform one element to the other? In other words, what is the shortest path connecting them on the Cayley graph? This problem reduces to computing the distance between vertices in the corresponding Cayley graph.
For $g,h\in G$, define the distance between them as the length of the shortest path connecting them on the Cayley graph, and let $\delta_S(g,h)$ denote their distance in $C_G$. By definition,
\begin{equation}
\delta_S(g,h)=\ell_S(hg^{-1}),
\end{equation}
and, in particular, if $e$ denotes the identity element of $G$,
\begin{equation}
\delta_S(g):=\delta_S(e,g)=\ell_S(g).
\end{equation}

For CNOT synthesis, we take $G=\GL(n,2)$ and
\begin{equation}
S=\mathcal T_n=\{T_{i,j}\mid i,j\in[n],\ i\neq j\},
\end{equation}
where each elementary transvection $T_{i,j}$ corresponds to one CNOT gate. Since $T_{i,j}^{-1}=T_{i,j}$, $C_G$ can be viewed as an undirected graph. An invertible parity matrix $A$ is therefore a vertex of $C_G$, a CNOT gate corresponds to an edge, and a CNOT circuit corresponds to a path. Consequently, computing the distance $\delta_{\mathcal T_n}(A_1,A_2)$ on the Cayley graph $C_{\GL(n,2)}$ corresponds to the exact CNOT synthesis of $A_2A_1^{-1}$. The complexity of exact CNOT synthesis therefore directly translates to the complexity of computing distances in $C_G$.

\begin{corollary}
Let $G=\GL(n,2)$ and let $S=\mathcal T_{n}$ be the set of elementary transvections. Given $A,B\in G$ and $k\in\mathbb N$, deciding whether $\delta_S(A,B)\le k$
is \NP-complete. Consequently, computing the exact distance $\delta_S(A,B)$ in $C_G$ is \NP-hard.
\end{corollary}

\subsection{Sequential XOR Programs}

Another application of our hardness result concerns XOR-program optimization in linear cryptography. An \emph{XOR program} is a sequence of instructions that perform computation using only XOR operations. Since XOR is equivalent to addition over $\Ftwo$, every value produced by such a program is a linear combination of the input bits. Consequently, XOR programs provide a natural implementation model for binary linear transformations. In symmetric cryptography, such transformations appear prominently as linear layers, which are important building blocks in many linear cryptographic protocols. Their implementation cost is therefore closely related to the number of XOR operations required, and the optimization of XOR programs has thereby been extensively studied~\cite{jean2017optimizing,duval2018mds,xiang2020optimizing,yuan2024framework}.

Several XOR-program models have been considered, differing mainly in how intermediate values are stored and reused. Two particularly relevant models are general XOR (g-XOR) and sequential XOR (s-XOR) programs~\cite{xiang2020optimizing,yuan2024framework}, which are illustrated in Fig.~\ref{fig:g-s-xor}. In a g-XOR program, an instruction may create a fresh intermediate variable, for example
\begin{equation}
y_k\leftarrow z_i\oplus z_j,\qquad i\neq j,
\end{equation}
where $y_k$ is a fresh variable and each of $z_i,z_j$ is either an input variable $x_r$ or a previously created variable $y_s$; the resulting value can subsequently be reused in later instructions. In contrast, an s-XOR program uses only in-place updates of the form
\begin{equation}
x_i\leftarrow x_i\oplus x_j,\qquad i\neq j,
\end{equation}
so the result overwrites one of the existing variables and no additional intermediate variable is introduced. For an invertible linear layer $A\in\GL(n,2)$, an s-XOR program starts from the $n$ input variables and successively updates them until they realize the outputs specified by $A$. We denote the minimum number of such instructions required to implement the target linear layer by $\#\mathsf{sXOR}(A)$.

\begin{figure}[pos=t]
  \centering
  \begin{subfigure}[t]{0.47\linewidth}
    \centering
    \begin{minipage}[c][4.5\baselineskip][c]{\linewidth}
    \centering
    {\normalsize
    $\begin{array}{@{}c@{\qquad}c@{}}
      \begin{aligned}
        y_1&\leftarrow x_1\oplus x_2\\
        y_2&\leftarrow y_1\oplus x_3\\
        y_3&\leftarrow y_2\oplus x_4
      \end{aligned}
      &
      \begin{aligned}
        y_4&\leftarrow y_1\oplus x_5\\
        y_5&\leftarrow y_4\oplus x_6\\
        y_6&\leftarrow y_5\oplus x_7
      \end{aligned}
    \end{array}$}
    \end{minipage}
    \caption{g-XOR program}
    \label{fig:g-xor-program}
  \end{subfigure}
  \hfill
  \begin{subfigure}[t]{0.47\linewidth}
    \centering
    \begin{minipage}[c][4.5\baselineskip][c]{\linewidth}
    \centering
    {\normalsize
    $\begin{array}{@{}c@{\qquad}c@{}}
      \begin{aligned}
        x_4&\leftarrow x_4\oplus x_1\\
        x_4&\leftarrow x_4\oplus x_2\\
        x_4&\leftarrow x_4\oplus x_3\\
        x_7&\leftarrow x_7\oplus x_1
      \end{aligned}
      &
      \begin{aligned}
        x_7&\leftarrow x_7\oplus x_2\\
        x_7&\leftarrow x_7\oplus x_5\\
        x_7&\leftarrow x_7\oplus x_6
      \end{aligned}
    \end{array}$}
    \end{minipage}
    \caption{s-XOR program}
    \label{fig:s-xor-program}
  \end{subfigure}
  \caption{Examples of g-XOR and s-XOR programs. Both XOR programs implement a  two-output linear map, which outputs $f_1(\v x)=x_1\oplus x_2\oplus x_3\oplus x_4$ and $f_2(\v x)=x_1\oplus x_2\oplus x_5\oplus x_6\oplus x_7$ on $y_3, y_6$ and $x_4,x_7$, respectively. The g-XOR program shares the intermediate value $x_1\oplus x_2$ and uses six instructions, whereas the s-XOR program updates the two target variables separately and uses seven.}
  \label{fig:g-s-xor}
\end{figure}

It is immediate that s-XOR programs and CNOT circuits are closely related. Each s-XOR instruction $x_i\leftarrow x_i\oplus x_j$ corresponds exactly to a CNOT gate with control $j$ and target $i$, and vice versa. Therefore, the correspondence is gate-by-gate: under the fixed-variable model, an s-XOR program is precisely the classical counterpart of a CNOT circuit. Consequently, for every $A\in\GL(n,2)$, we have $\#\mathsf{sXOR}(A)=\#\cnot(A)$.

From the perspective of computational complexity, the minimization of g-XOR programs has been shown to be \NP-hard by Boyar, Matthews, and Peralta~\cite{boyar2013logic}; this result, however, exploits a model in which arbitrary intermediate variables can be introduced. The complexity of the more restrictive in-place s-XOR optimization problem was not resolved by that result; subsequent works instead developed exhaustive-search and heuristic methods for minimizing s-XOR count~\cite{jean2017optimizing,xiang2020optimizing,yuan2024framework}.
The correspondence between CNOT circuits and s-XOR programs then allows our hardness result to fill this gap
directly: the hardness of vanilla exact CNOT synthesis transfers immediately to
the optimization of fixed-width, in-place s-XOR programs.


\begin{corollary}
Given $A\in\GL(n,2)$ and $k\in\mathbb N$, deciding whether $A$ admits an s-XOR program implementation of length at most $k$ is \NP-complete. Consequently, computing $\#\mathsf{sXOR}(A)$ is \NP-hard.
\end{corollary}

\subsection{Exact Phase Polynomial Synthesis}\label{sec:phase_polynomial}

Phase polynomial circuits, consisting of CNOT and $R_z$ gates, constitute an important and widely studied class of quantum circuits. This important circuit class arises in many quantum algorithms, including variational algorithms such as VQE~\cite{vqe} and QAOA~\cite{qaoa}. Here, the $Z$ rotation gate $R_z(\theta)$ is defined as
\begin{equation}
  R_z(\theta)|x\>=e^{2\pi i\theta x}|x\>.
\end{equation} In this section, we show that our hardness result and proof strategy can also be extended to the exact synthesis of phase polynomial circuits.

Following Amy et al.~\cite{t-par,amy2018CNOTcomplexity}, a circuit over CNOT and $Z$ rotation gates can be described by a pair $(f,A)$, where $A\in\GL(n,2)$ specifies the target linear transformation on the computational basis states and $f(\v x)$ is the phase polynomial. In particular, the circuit is then required to implement the following unitary:
\begin{equation}
  U_{(f,A)}=\sum_{\v x\in\Ftwo^n} e^{2\pi i f(\v x)}|A\v x\>\<\v x|.
\end{equation}
This representation is also called \emph{sum-over-path} representation in the literature.
In particular, the phase polynomial $f$ is usually expressed in the following form:
\begin{equation}
f(\v x)=\sum_{\v a\in\Ftwo^n}\widehat f(\v a)\chi_{\v a}(\v x),\qquad \chi_{\v a}(\v x)=\bigoplus_{j=1}^{n}a_jx_j.
\end{equation}
where $\widehat f(\v a)$ are the phases and $\chi_{\v a}(\v x)$ are the parities associated with the phases. Thus, the synthesis problem contains both a linear target, specified by $A$, and a phase component, specified by the phases and parities appearing in $f(\v x)$. We denote the set of all the parities $\chi_{\v a}(\v x)$ that appear in $f$  by $\mathrm{supp}(f)$.

Synthesizing phase polynomial circuits is slightly more complicated than synthesizing CNOT circuits, as two types of parity constraints are required: besides realizing the prescribed final linear transformation $A$, the CNOT circuit skeleton must also reach all the intermediate parities $\mathrm{supp}(f)$ appearing in the phase polynomial $f(\v x)$. These intermediate parities serve as checkpoints: once a required parity appears on a qubit, the corresponding phase gate ($R_z$ gate) can then be inserted at the location to realize the associated phase term. Thus, phase polynomial synthesis can be viewed as a two-stage process: first constructing a CNOT circuit skeleton that reaches a set of prescribed parity checkpoints while ending at the prescribed linear transformation $A$, and then inserting phase gates at the parity checkpoints, where, the second stage is relatively easy and can be done within polynomial time.

Consequently, exact synthesis of CNOT circuits can be naturally reduced to the (CNOT-minimal) exact synthesis of phase polynomial circuits, and the latter problem is at least as hard as the former:

\begin{corollary}\label{corollary:f_trivial}
Exact synthesis of phase polynomial circuits is \NP-hard, even when the phase polynomial is trivial ($\mathrm{supp}(f)=\varnothing$).
\end{corollary}

Amy et al.~\cite{amy2018CNOTcomplexity} established hardness for a restricted phase polynomial synthesis setting through the \emph{fixed-target parity-network} model, where all CNOT gates must target a common qubit and all the required parity checkpoints are synthesized on that specific target qubit. In this fixed-target setting, the synthesis problem is essentially equivalent to a traveling salesman problem over the required parities on the target qubit, making the corresponding hardness proof relatively straightforward. Our result removes this restriction and establishes hardness for the general setting: exact synthesis of phase polynomial circuits is already \NP-hard, even when the phase polynomial is trivial, $f(\v x)=0$.

This reduction, however, deliberately places all of the hardness in the target linear transformation $A$ and does not exploit the intermediate parity requirements introduced by a nontrivial phase polynomial. A natural question is: do the checkpoint constraints themselves intrinsically contribute to the complexity of exact phase polynomial synthesis? One might ask, in the extremal case where the target linear transformation is fixed to the identity, can we still prove that exact phase polynomial synthesis is \NP-hard? The answer to this extremal condition complexity problem is \emph{yes}. We next show that by adapting the reduction strategy for our main result Theorem~\ref{thm:main}, we can prove the following hardness result within the same analysis framework.
\begin{theorem}\label{thm:checkpoint_identity}
  Exact synthesis of phase polynomial circuits is \NP-hard, even when the target linear transformation is trivial ($A=I_n$).
\end{theorem}

The proof is based on a reduction from the Hamiltonian \emph{cycle} problem on grid graphs, which we abbreviate as \textsf{Grid-HC}. Recall that \textsf{Grid-HC} was also  shown to be \NP-complete by Itai et al.~\cite[Theorem~2.1]{itai1982hamilton}. Our proof follows the strategy of Theorem~\ref{thm:main}, but only uses the path and coordinate qubits. Recorder qubits are unnecessary here; this simplification is possible because we can now adjust the specified phase polynomial to impose parity checkpoint constraints.

\begin{proof}
 Let $G=(V_G,E_G)$ be an input grid graph with $m=|V_G|$. If $\delta(G):=\min_{v_G\in V_G}\deg(v_G)<2$ or $|V_G|\le2$, the graph cannot contain a Hamiltonian cycle. In this case, the problem instance is trivial and can be immediately determined as negative without reduction.

 Otherwise, choose any vertex $\v s\in V_G$ as the anchor point. Then apply the same encoding scheme in Sec.~\ref{sec:reduction} to embed $G$ into a hypercube. In particular, we have
\begin{equation}
  V=\{\mathrm{enc}(\v v_G)\oplus\mathrm{enc}(\v s):\v v_G\in V_G\}
  \subseteq\Ftwo^d
\end{equation}
and the corresponding induced hypercube subgraph is $Q_d[V]$. Then we have $\v 0\in V$ and  $d\le m-1$. We use one path qubit and $d$ coordinate qubits with input variables $\v x=(x_0,x_1,\ldots,x_d)$. We then construct a corresponding phase polynomial synthesis problem instance as follows:
\begin{equation}
  f(\v x)=\sum_{\v v\in V}\frac13\overline\chi_{\v v}(\v x),\quad A=I_{d+1}, \quad k=m.\label{eq:phase_poly_reduction_ins}
\end{equation}
The bounded decision version of exact phase polynomial synthesis problem then asks \emph{whether there exists a $(d+1)$-qubit phase polynomial circuit that contains at most $k$ CNOT gates and implements $U_{(f,A)}$}. Here, the angles $\frac13$ are chosen to be nonzero and non-dyadic (cannot be written in the form $\frac{a}{2^b}$) so that all parities $\overline\chi_{\v v}(\v x)$ in $f$ must be reached, and cannot be skipped or eliminated via further optimization (see~\cite[Proposition 2.7]{amy2018CNOTcomplexity}). The construction is illustrated in Fig.~\ref{fig:phase_identity_reduction}. Next, we show that
\begin{multline}
   Q_d[V]\text{ admits a Hamiltonian cycle}
    \quad\Longleftrightarrow\\
    U_{(f,A)}\text{ can be implemented with }\mathrm{CNOT}+R_z\text{ circuit containing at most }k \text{ CNOT gates}.
\end{multline}

\begin{figure}[pos=t]
  \centering
  \resizebox{\linewidth}{!}{%
  \begin{tikzpicture}[font=\footnotesize]
    \begin{yquant}[operator/separation=0.4cm,register/separation=0.3cm]
      qubit {$x_0$} path;

      qubit {$x_1$} c1;
      qubit {$x_2$} c2;
      qubit {\raisebox{1mm}{\(\vdots\)}} cdots;
      setstyle {draw=none} cdots;
      qubit {$x_d$} cd;

      [name=phasezero]
      box {$R_z(\frac13)$} path;

      [name=phasewalkone]
      cnot path | c1;

      [name=phasevtwo]
      box {$R_z(\frac13)$} path;

      [name=phasewalkd]
      cnot path | cd;

      [name=phasevthree]
      box {$R_z(\frac13)$} path;

      align path,c1,c2,cdots,cd;
      hspace {3.2cm} path,c1,c2,cdots,cd;

      [name=phasepathout]
      output {$x_0$} path;
      [name=phasec1out]
      output {$x_1$} c1;
      [name=phasec2out]
      output {$x_2$} c2;
      [name=phasecdotsout]
      output {\raisebox{1mm}{\(\vdots\)}} cdots;
      [name=phasecdout]
      output {$x_d$} cd;
    \end{yquant}

    \coordinate (phasecircuitwest) at (current bounding box.west);
    \coordinate (phasebracecolumn) at ([xshift=0mm]phasecircuitwest);
    \coordinate (phaseannotationright) at ([xshift=-1.5mm]phasebracecolumn);
    \path (phasewalkone-p) -- (phasewalkd-p)
      coordinate[midway] (phasecoordinatemid);

    \node[anchor=east,text width=25mm,inner sep=0pt,
          align=center,font=\footnotesize]
      at ([xshift=3.5mm]phaseannotationright |- phasevtwo)
      {Path Qubit};

    \draw[decorate,decoration={brace,mirror,amplitude=4pt}]
      (phasebracecolumn |- phasewalkone-p)
      -- (phasebracecolumn |- phasewalkd-p);
    \node[anchor=east,text width=25mm,inner sep=0pt,
          align=center,font=\footnotesize]
      at ([xshift=3.5mm]phaseannotationright |- phasecoordinatemid)
      {Coordinate\\Qubits};

    \node[anchor=south,align=center]
      at ([yshift=0mm]phasezero.north)
      {$\overline\chi_{\v 0}(\v x)$};
    \node[anchor=south,align=center]
      at ([yshift=0mm]phasevtwo.north)
      {$\overline\chi_{\v w_2}(\v x)$};
    \node[anchor=south,align=center]
      at ([yshift=0mm]phasevthree.north)
      {$\overline\chi_{\v w_3}(\v x)$};

    \node[inner sep=0pt]
      at (phasecdotsout -| phasewalkone-p) {\raisebox{1mm}{\(\vdots\)}};
    \path (phasevthree.east) -- (phasepathout.west)
      coordinate[midway] (phasemoregates);
    \node[inner sep=0pt,fill=white]
      at (phasemoregates) {$\cdots$};
    \node[inner sep=0pt,fill=white]
      at (phasec1out -| phasemoregates) {$\cdots$};
    \node[inner sep=0pt,fill=white]
      at (phasec2out -| phasemoregates) {$\cdots$};
    \node[inner sep=0pt,fill=white]
      at (phasecdotsout -| phasemoregates) {\raisebox{1mm}{\(\vdots\)}};
    \node[inner sep=0pt,fill=white]
      at (phasecdout -| phasemoregates) {$\cdots$};
  \end{tikzpicture}%
  }
  \caption{Demonstration of the reduction from hypercube Hamiltonian cycle to exact phase polynomial synthesis with $A=I_{d+1}$. For illustration, the first two nontrivial parity checkpoints are set to $\overline \chi_{\v w_2}(\v x)=x_0\oplus x_1$ and $\overline \chi_{\v w_3}(\v x)=x_0\oplus x_1\oplus x_d$. Each $R_z(\frac13)$ gate is applied when a required parity appears on the path qubit. The omitted gates complete the cycle and restore the parity on the path qubit to $x_0$.}
  \Description{One path qubit and a bundle of coordinate qubits. Phase rotations are applied to the path qubit at successive required parities, with coordinate-controlled CNOT gates between them. The omitted continuation returns every qubit to its input parity.}
  \label{fig:phase_identity_reduction}
\end{figure}
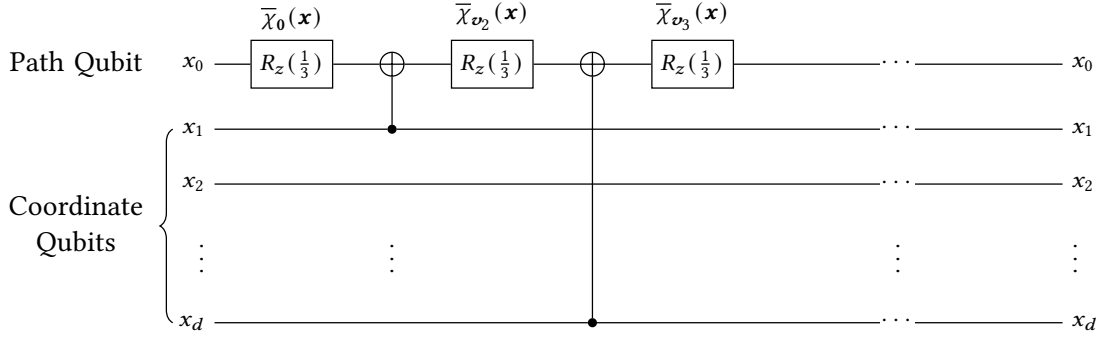

\paragraph{\textbf{Completeness} ($\Rightarrow$): Constructing a phase polynomial circuit from a Hamiltonian cycle.} Suppose $Q_d[V]$ has a Hamiltonian cycle
\begin{equation}
  \v w_1(=\v 0)-\v w_2-\cdots-\v w_m-\v w_{1}(=\v 0).
\end{equation}
For every consecutive pair, the Hamiltonian cycle edge implies that they differ in exactly one coordinate, say $\v w_i\oplus\v w_{i+1}=\v e_\ell$. A CNOT controlled by coordinate qubit $\ell$ and targeting the path qubit then updates
\begin{equation}
  \overline\chi_{\v w_i}(\v x) \longmapsto \overline\chi_{\v w_i}(\v x)\oplus x_\ell=\overline\chi_{\v w_{i+1}}(\v x)
\end{equation}
by Eq.~\eqref{eq:hypercube_parity_adj}. The $m$ edges therefore give exactly $m=k$ CNOT gates. Among the circuit, all checkpoints are reached. At the end of the circuit, the path qubit returns to $x_0$ and every coordinate qubit remains unchanged, implying that the final parity matrix is $I_{d+1}$. As a last step, insert an $R_z(\frac13)$ gate at every location that parity $\overline\chi_{\v v}(\v x)$ appears, then we have derived a phase polynomial circuit that implements the sum-over-path specification $(f,A)$ in~\eqref{eq:phase_poly_reduction_ins}.

\paragraph{\textbf{Soundness} ($\Leftarrow$): Extracting a Hamiltonian cycle from a short phase polynomial circuit.} Conversely, consider an arbitrary CNOT+$R_z$ circuit that contains $L:=|C|\le m$ CNOT gates and satisfies the sum-over-path specification~\eqref{eq:phase_poly_reduction_ins}. We do not assume that the circuit admits the path and coordinate qubits distinction;  instead, we show that the tight gate budget forces every such circuit to exhibit this structure. Throughout the subsequent analysis, we remove all the $R_z$ gates, and leave only the CNOT circuit skeleton.

First, we show that \emph{all CNOT gates must target the same qubit}. To begin with, Proposition~2.7 of~\cite{amy2018CNOTcomplexity} ensures that the CNOT skeleton reaches every parity $\overline\chi_{\v v}(\v x)$ with $\v v\in V$. Except for $\overline\chi_{\v 0}(\v x)=x_0$, none of these parities is present initially; so generating the remaining $m-1$ parities requires at least $m-1$ CNOT gates.

Let $b$ be the number of qubits targeted by at least one CNOT gate. For each such qubit $j$ ($0\le j\le d$), consider the last gate targeting it. Since the final linear transformation is $A=I_{d+1}$, this gate must restore its original input variable $x_j$ on qubit $j$. None of these $b$ gates can be among the $m-1$ gates counted above, as they leave the original variable $x_j$ and do not generate any new parity $\overline \chi_{\v v}(\v x)$. Hence
\begin{equation}
  m-1+b\le L\le m.
  \label{eq:phase_identity_accounting}
\end{equation}
Apparently $b\ge1$, so it follows that $b=1$ and $L=m$. Thus every CNOT gate targets the same qubit.

We next show that \emph{this unique target qubit must be the path qubit}. The anchor vertex $\v s$ has at least two distinct neighbors, say $\v e_p,\v e_q\in V$. Then $x_0\oplus x_p$ and $x_0\oplus x_q$ both lie in the support of $f$. Suppose that the unique target were coordinate qubit $j>0$, so other qubits are never targeted by CNOT and remain unchanged. As a result, the coefficient of $x_j$ on the target qubit remains one throughout the circuit. However, during the traversal of all parities $\overline\chi_{\v v}(\v x)$ on the target qubit, at least one of $x_0\oplus x_p$ and $x_0\oplus x_q$ has coefficient zero on $x_j$, resulting in a contradiction. Therefore the unique target must be the path qubit.

To summarize, all CNOT gates therefore target the path qubit and are controlled by coordinate qubits. Of the $m$ CNOT gates, $m-1$ gates generate the $m-1$ non-initial parities, and the final gate restores $x_0$ on the path qubit. In their order of occurrence on the path qubit, these parities form a closed sequence in which consecutive terms differ by one coordinate variable. By Eq.~\eqref{eq:hypercube_parity_adj}, the corresponding vertices form a cycle that visits every vertex of $V$ exactly once before returning to $\v 0$. This is a Hamiltonian cycle in $Q_d[V]$, and hence there is a Hamiltonian cycle in $G$.

Now we have established that the bounded decision version of phase polynomial circuit synthesis is \NP-hard\footnote{The bounded decision problem is actually \NP-complete, as its instances can be efficiently verified via a short phase polynomial circuit witness. But it suffices to show the decision problem \NP-hard in order to prove the \NP-hardness of the corresponding synthesis problem.}. As a corollary, the exact \emph{synthesis} problem, which, given a sum-over-path specification $(f,A)$, asks for a circuit implementing $U_{(f,A)}$ with the minimum possible number of CNOT gates, is also \NP-hard, even when $A$ is the identity.
\end{proof}

Theorem~\ref{thm:checkpoint_identity} shows that intermediate checkpoint constraints alone can make CNOT synthesis intractable, even when the target linear transformation is trivial. This complements Corollary~\ref{corollary:f_trivial}  which sets $f=0$ and places the hardness entirely in the final matrix $A$. The two arguments identify separate sources of difficulty in exact phase polynomial synthesis: both the target linear transformation constraint and the phase polynomial constraint alone is enough to make the exact synthesis problem intractable. This result also highlights the broader applicability of our reduction framework: the same approach based on Hamiltonian  path/cycle problems, isomorphic embeddings and hypercubes extends beyond vanilla exact CNOT synthesis to related circuit synthesis problems.

From a technical perspective, our proof recovers the fixed-target circuit structure underlying hardness result by Amy et al.~\cite{amy2018CNOTcomplexity}: all CNOT gates must target the same qubit. Amy et al. added this condition in the problem definition to derive the hardness result, whereas our proof does not require the assumption, but enforces the fixed-target circuit structure with the carefully designed tight gate count. Thereby we have dropped the fixed-target condition of circuits in proving the hardness of exact phase polynomial synthesis.

\section{Related Works}

\paragraph{Computational complexity of circuit synthesis and optimization.} Amy et al.~\cite{amy2018CNOTcomplexity} proved \NP-completeness for fixed-target parity networks and parity networks with encoded inputs, while Kang and Ma~\cite{kang2023cnot} established \NP-hardness for ancilla-free CNOT synthesis under restricted topology. Jiang et al.~\cite{jiang2020optimal} further proved strong inapproximability for topology-constrained CNOT synthesis with clean ancillas and for local size minimization: approximating either problem within any constant factor is \NP-hard.

For the Clifford+$T$ gate set, van de Wetering and Amy~\cite{vandewetering2024optimising} proved that minimizing the $T$ count or depth and the CNOT or Hadamard count of Clifford+$T$ circuits is \NP-hard under polynomial-time Turing reductions; their argument also covers Toffoli-count minimization in classical reversible circuits.

Classical logic synthesis likewise contains many hard circuit-minimization problems, echoing our result. Boyar et al.~\cite{boyar2013logic} proved \NP-hardness of minimizing general XOR straight-line programs, while Parberry showed that sorting-network verification is \textsf{coNP}-complete even near optimal depth~\cite{parberry1991sorting}. Correspondingly, SAT is widely used for exact logic and sorting-network synthesis~\cite{haaswijk2020SATBased,bundala2017sorting}.

\paragraph{Synthesis and optimization of CNOT and phase polynomial circuits.} CNOT and phase polynomial synthesis have been studied extensively; representative elimination, greedy, and heuristic methods were summarized in the introduction~\cite{pmh_optimal_linear_synthesis,debrugiere2021Gaussian,amy2018CNOTcomplexity,chen2025phasepoly}. Shaik and van de Pol encoded exact CNOT synthesis into SAT, classical planning, and quantified Boolean formula (QBF) instances and solved them with existing solvers; the SAT-based method was the most effective overall~\cite{shaik2024Optimal}. Qiskit-SAT similarly encodes exact CNOT synthesis as SAT and uses Z3 to solve the resulting instances~\cite{qiskit-sat,Z3}.

More recently, methods for exact CNOT synthesis other than SAT have emerged. Webster et al.~\cite{CNOT_opt_UCL} studied heuristic and exact synthesis for CNOT and Clifford circuits; their exact synthesis uses pre-computed databases and does not scale at runtime. Lin-search instead scales exact CNOT synthesis through hybrid iterative deepening search~\cite{li2026linsearch}. Li et al.~\cite{li2026parallelizable} developed a parallel STP-based framework for the exact synthesis of both CNOT and phase polynomial circuits.

HOPPS subsequently introduced, to our knowledge, the first dedicated hardware-aware exact phase polynomial synthesizer for arbitrary device topologies~\cite{li2025HOPPS}. The \NP-hardness result in Section~\ref{sec:phase_polynomial} provides a complexity-theoretic foundation for these exact methods, even without topology constraints or a nontrivial phase polynomial.

\paragraph{Concurrent work.} While validating the proof and preparing this manuscript, we became aware of independent concurrent work by Acuaviva et al.~\cite{acuaviva2026cnot}, which independently proves \NP-completeness of all-to-all, fixed-label CNOT distance. Their reduction is from \textsf{VertexCover}, following a source problem also used in earlier CNOT and XOR-program hardness proofs~\cite{kang2023cnot,boyar2013logic}, whereas ours introduces a distinct reduction from \GridHP{} bridged from an isometric hypercube embedding.

The two reductions expose complementary sources of difficulty. In the construction of Acuaviva et al., vertex selection represents which intermediate parities should be computed and shared; their lower bound extracts this structure from an unrestricted circuit through XOR-DAG projection and contraction. Our proof instead exposes the difficulty of ordering intermediate parities: recorder qubits and a rank argument force a single trajectory through the required hypercube vertices.


\section{Conclusion}

In this work, we established the \NP-completeness of vanilla exact CNOT synthesis: given an invertible parity matrix $A\in\GL(n,2)$ and a gate bound $k$, deciding whether $A$ can be implemented by at most $k$ CNOT gates is \NP-complete even with all-to-all connectivity, no ancillas, fixed labeled qubits, and identity input. Our reduction proceeds from \GridHP{} through an induced hypercube representation, exploiting the direct correspondence between hypercube moves and parity updates on CNOT circuits. The main difficulty is that CNOT synthesis constrains only the final parity matrix, whereas a Hamiltonian path imposes requirements on the intermediate states. Recorder qubits, replication, and a tight gate budget bridge this gap and force every sufficiently short implementation to expose the required Hamiltonian path structure. The result therefore shows that the combinatorial complexity of in-place linear reversible synthesis alone is sufficient to make exact CNOT optimization intractable, without relying on topology restrictions, ancillary workspace, or encoded inputs.

Because the vanilla problem has a particularly simple algebraic formulation, the result also transfers directly to several different but important optimization problems. In particular, it establishes \NP-completeness of (1) the shortest word problem over $\GL(n,2)$ generated by elementary transvections, (2) the distance problem on Cayley graphs, (3) the \NP-hardness of optimal fixed-width s-XOR programs and (4) exact synthesis of phase polynomial circuits. For phase polynomial circuits, our hardness results reveal two complementary yet independently sufficient sources of hardness: both the target linear transformation constraint and the phase polynomial constraint contribute to the hardness of exact phase polynomial synthesis, and either constraint alone is enough make the problem \NP-hard.

\section*{Acknowledgments}

This work was partially supported by the Quantum Science and Technology-National Science and Technology Major Project (Grant No. 2024ZD0300500),  the Youth Innovation Promotion Association of Chinese Academy of Sciences (Grant No. 2023116) and the National Natural Science Foundation of China (Grant No. 62402485).

\section*{AI Tool Usage Statement}

The authors used GPT-5.6 in the development of this work. All mathematical proofs, analysis, citations and the final text were reviewed and validated by the authors, who take full responsibility for the final work.

\printcredits

\bibliographystyle{elsarticle-num}

\bibliography{libpaper}

\end{document}